\documentclass[journal,onecolumn,letterpaper]{IEEEtran}

\usepackage[utf8]{inputenc} 
\usepackage[T1]{fontenc}
\usepackage{url}
\usepackage{ifthen}
\usepackage{cite}
\usepackage{caption}
\usepackage[cmex10]{amsmath} % Use the [cmex10] option to ensure complicance
\usepackage[pdftex]{graphicx}
\usepackage{xcolor, tikz, hyperref}
\usepackage{amsmath,amssymb,amsthm}
\usepackage{tikz,enumerate}
\usepackage{mathtools}
\usepackage{etex,mathrsfs}
\usepackage{tikz}  
\usetikzlibrary{positioning} % For easier positioning of nodes  
\usetikzlibrary{calc}
\usetikzlibrary{decorations.markings}

\theoremstyle{plain}

\newtheorem{theorem}{Theorem}
\newtheorem{proposition}{Proposition}
\newtheorem{corollary}{Corollary}

\newtheorem{lemma}{Lemma}

\theoremstyle{remark}

\newtheorem{remark}{Remark}

\newtheorem{example}{Example}

\theoremstyle{definition}
\newtheorem{definition}{Definition}
\begin{document}
\title{A New Upper Bound for Distributed Hypothesis Testing Using the Auxiliary Receiver Approach} 

% %%% Single author, or several authors with same affiliation:
\author{
 \IEEEauthorblockN{Zhenduo Wen and Amin Gohari\\}
\IEEEauthorblockA{Department of Information Engineering\\
                   The Chinese University of Hong Kong\\
                   ShaTin, NT, Hong Kong SAR\\
                   Email: \{wz124,agohari\}@ie.cuhk.edu.hk}
                   }

\maketitle

\begin{abstract}
   This paper employs the add-and-subtract technique of the auxiliary receiver approach to establish a new upper bound for the distributed hypothesis testing problem of distinguishing between the two hypotheses $P_{XY}$ and $Q_{XY}$ using one-way communication. The new bound has fewer assumptions than the upper bound proposed by Rahman and Wagner, is at least as tight as the bound by Rahman and Wagner, and outperforms it in certain discrete and Gaussian settings. Conceptually speaking, unlike Rahman and Wagner, who view their additional receiver as \emph{side information}, we view it as an auxiliary receiver and use a different manipulation for single-letterization.
The new bound also implies that the minimum rate $R$ needed to achieve the centralized bound is at least the conditional entropy $H_P(X|Y)$ when $P_{XY}$ and $Q_{XY}$ are non-degenerate distributions.   

\end{abstract}

\section{Introduction}

%% Note that \label must occur AFTER (or within) \caption.
%% For figures, \caption should occur after the \includegraphics.
%%

This paper \footnote{This work "A New Upper Bound for Distributed Hypothesis Testing Using the Auxiliary Receiver Approach" was partially presented at the 2025 IEEE International Symposium on Information Theory.} concerns the problem of hypothesis testing in a distributed setting. We consider the one-sided two-terminal setup as described in Fig.~\ref{fig1}. Let \(X\) and \(Y\) be two sources with discrete alphabets \(\mathcal{X}\) and \(\mathcal{Y}\), respectively. Two terminals observe independent and identically distributed (i.i.d.) sequences \(X^n\) and \(Y^n\) over \(n\) time instances. Under the null hypothesis \(\mathcal{H}_0\), the joint distribution is \(P_{XY}\), while under the alternative hypothesis \(\mathcal{H}_1\), it is \(Q_{XY}\).
The observed sequence of \(X\) is transmitted over an \((n, R)\) channel using an encoding function \(f_n: \mathcal{X}^n \rightarrow \mathcal{M} = \{1, 2, \ldots, 2^{\lfloor nR \rfloor}\}\). The detector receives the message \(m = f_n(x^n)\) and decides between the hypotheses \(\mathcal{H}_0\) and \(\mathcal{H}_1\) based on the following discriminant rule:
\begin{align*}
g_n(m, y^n) = 
\begin{cases}
\mathcal{H}_0, & \text{if } (m, y^n) \in \mathcal{A} \\ 
\mathcal{H}_1, & \text{otherwise}
\end{cases}
\end{align*}
where \(\mathcal{A} \subseteq \mathcal{M} \times \mathcal{Y}^n\) is the acceptance region for \(\mathcal{H}_0\). See Fig.~\ref{fig1} for an illustration.

The performance of a Distributed Hypothesis Test (DHT) can be quantified by the error exponent with respect to the communication rate \(R\). Let \(p(\epsilon, n, R)\) denote the smallest achievable error probability under the alternative hypothesis \(\mathcal{H}_1\) (Type II error) when the source blocklength is \(n\), the message length does not exceed \(\lfloor nR \rfloor\) bits, and the error probability under the null hypothesis \(\mathcal{H}_0\) (Type I error) is upper bounded by \(\epsilon\). The error exponent of \(p(\epsilon, n, R)\) is defined as
\[
E_{P_{XY}, Q_{XY}}(R) \overset{\triangle}{=} \lim_{\epsilon \rightarrow 0} \lim_{n \rightarrow \infty} -\frac{1}{n} \log p(\epsilon, n, R).
\]
\subsection{Literature Review}
\begin{figure}[htbp]% [htbp] are placement options  
\centering % Center the figure (and the tikzpicture within it) 
\resizebox{0.6\linewidth}{!}{
\begin{tikzpicture}

% Draw blocks  
\node (X) {$X^n$};
\node[right=of X, draw, rectangle] (Enc) {$\begin{aligned}
    & ~~~\text{Encoder}\\
    & \ m=f_n(x^n)
\end{aligned}$};
% \node[right= of Enc] (M) {$M$};
\node[right= of Enc, draw, xshift=1cm, rectangle] (Dec) {$\begin{aligned}
    & ~\text{Detector}\\ &g_n(m, y^n)
\end{aligned}$};
\node[below= of Dec] (Y) {$Y^n$};
\node[right= of Dec] (H) {$\mathcal{H}_0 / \mathcal{H}_1$};
% Using the positioning library  
  
% Draw an arrow from X to Y  
\draw[->] (X) -- (Enc); 
\draw[->] (Enc) -- (Dec) node[midway, above] {$M$};
% \draw[->] (M) -- (Dec);
\draw[->] (Y) -- (Dec);
\draw[->] (Dec) -- (H);

% Draw a dotted curve from X to Y
\draw[dotted, thick, color=black] (X) to[bend right] node[midway, above, yshift=5pt] {$P_{XY}/ Q_{XY}$} (Y);
\end{tikzpicture}  
}
\caption{The problem setup}
\label{fig1}
\end{figure}
The following multi-letter characterization of $E_{P_{XY},Q_{XY}}(R)$ is known:

\begin{theorem}[Theorems 4 in \cite{AhlswedeCsiszar1986}]\label{thm14Ahlswede} We have
$$
E_{P_{XY},Q_{XY}}(R) =\sup_{n, P_{M|X^n}}\frac{1}{n} D(P_{MY^n} \| Q_{MY^n}),
$$
where the message $M$ takes values in $\mathcal{M} = \{1,2, \cdots, 2^{\lfloor nR \rfloor}\}$.\label{thm1cz}
\end{theorem}

Quantization-based lower bounds were obtained by Ahlswede and Csiszar \cite{AhlswedeCsiszar1986}:

\begin{theorem}[Theorem 5 in \cite{AhlswedeCsiszar1986}]\label{thm5Ahlswede}
For any \(R > 0\), we have
\begin{align*}
    &E_{P_{XY},Q_{XY}}(R) \geq \max_{P_{U|X}:\ I_P(U;X) \leq R} \left [ D(P_X \| Q_X) + D(P_{UY} \| Q_{UY}) \right ].
\end{align*}
\end{theorem}

The use of \(P_{U|X}\) implies a Markov chain \(U \rightarrow X \rightarrow Y\). The lower bound is not generally tight, as mentioned in \cite[Section 3]{AhlswedeCsiszar1986}. The work is then extended and improved by Han \cite{TeHan1987}, and Shimokawa, Han, and Amari (SHA) using the quantization-and-binning technique \cite{Shimokawa1994}. Further improvements following this line include \cite{Haim2016binary}, \cite{weinberger2019reliability}, \cite{Watanabe2022sub}, \cite{Kochman2023improved}, and \cite{kochman2021communication}. 

In terms of the upper bound, the following facts and results are known: The upper bound given by the Neyman-Pearson Lemma in non-distributed hypothesis testing provides the following trivial but useful benchmark at high rates. This bound is called the centralized bound:
\begin{align}
E_{P_{XY},Q_{XY}}(R)\leq D(P_{XY}\|Q_{XY}). \label{centralized_bound}
\end{align}
The centralized bound is tight when $R\geq D(P_{XY}\|Q_{XY})+H_P(X|Y)$, see \cite{Shimokawa1994}. One of the results of this paper is that if the centralized bound is tight, we must have
$R\geq H_P(X|Y)$. Moreover, in some cases, we can provide better bounds than $H_P(X|Y)$. 

In addition to characterizing the optimal exponent at high rates when $R\geq D(P_{XY}|Q_{XY})+H_P(X|Y)$, the optimal exponent is also known in some other special settings. 
The authors in \cite{AhlswedeCsiszar1986} found the optimal exponent for the special case of testing against independence where \(Q_{XY} = P_X P_Y\). Rahman and Wagner \cite{RahmanWagner2011} expanded on this by finding the optimal exponent for the "testing against conditional independence" setting. At zero rate, when we allow for one bit of communication, Han \cite{TeHan1987} computed the optimal exponent.

In this paper, we are mainly concerned with upper bounds on the exponent $E_{P_{XY}, Q_{XY}}(R)$.
Authors in \cite{RahmanWagner2011}  obtained the following upper bound:

\begin{definition}[Eqns (C14) and (C15) in \cite{RahmanWagner2011}]\label{definition_1}
Let \(  \mathcal{R}(P_{XY}, Q_{XY})\) be the set of all \((P_{Z|XY}, Q_{Z|XY})\) such that
\begin{align}
      Q_{YZ|X} &= Q_{Z|X} Q_{Y|Z}, \label{eqnDef2an} \\
     P_{XZ} &= Q_{XZ}.\label{eqnDef2bn} 
\end{align} 
Let $\mathcal{S}(P_{XYZ})$ be the set of all $P_{U|X}$ such that $P_{UXYZ}=P_{U|X}P_{XYZ}$ satisfies
$$I_P(X;U|Z)\leq R.$$
\end{definition}

\begin{theorem}[Theorem 3, Corollary 5 in \cite{RahmanWagner2011}]
We have
\begin{align*}
     E_{P_{XY},Q_{XY}}(R) \leq &\min_{P_{Z|XY}, Q_{Z|XY} \in \mathcal{R}(P_{XY}, Q_{XY})} [ D(P_{YZ} \| Q_{YZ}) +  \max_{P_{U|X}\in\mathcal{S}(P_{XYZ})} I_P(Y; U | Z)].
\end{align*}
\label{RW_thm}
\end{theorem}
Rahman and Wagner interpret \( Z \) as side information by employing the following inequality:
\[
D(P_{MY^n} \| Q_{MY^n}) \leq D(P_{MY^nZ^n} \| Q_{MY^nZ^n})
\]
where $M$ is the transmitted message. In other words, receiver $Y$ is replaced with $Z'=(Y,Z)$.
In our derivation, we bypass this step and instead utilize a telescopic sum expansion. Consequently, \( Z \) does not function as side information but rather as a general auxiliary receiver.

Another upper bound by Hadar et al. \cite{hadar2019error} involves testing for correlation between Gaussian sequences using the Ornstein-Uhlenbeck process and optimal transport techniques. Notably, neither of the bounds in \cite{RahmanWagner2011} and \cite{hadar2019error} is superior to the other one across all parameters for the Gaussian problem.

A more recent improvement is obtained in \cite{Yuval2024} by giving the detector not just one but a series of side information (in the sense of \cite{RahmanWagner2011}). This idea is used in \cite{Yuval2024} to improve over the bound in \cite{RahmanWagner2011}. The idea of using multiple auxiliary receivers (instead of just one auxiliary receiver) is independent of the idea discussed in this paper; indeed, our approach can be combined with the idea in \cite{Yuval2024} to potentially develop even better bounds.

In this paper, we propose a computable upper bound on $E_{P_{XY},Q_{XY}}(R)$ using a systematic manipulation on the auxiliary receivers \cite{gon21}. The auxiliary receiver approach is a technique for writing new outer bounds for different communication problems. It has been used in \cite{gon21} to derive new bounds for the broadcast, relay, and interference channels. 
In this paper, we utilize the "add-and-subtract" technique of the auxiliary receiver approach to derive a new upper bound for the distributed hypothesis testing problem.\footnote{The \textit{message augmentation} technique of the auxiliary receiver approach appears to yield only the bound presented in Proposition \ref{thmJ}, which can be proven more readily through a direct method.}

 Finally, we note that the literature studies different variants and extensions of the problem considered here \cite{escamilla2018, salehkalaibar2020, weinberger2019b, watanabe2018, tian2008, salehkalaibar2018,xiang2013interactive,katz2016collaborative,Shalaby1992multiterminal,sadaf20,Escamilla20, han2006exponential, hamad2023multi, zaidi2022rate, inan2022fundamental, escamilla2020distributed, zhao2014distributed, zhao2018distributed, sreekumar2019distributed, sreekumar2024distributed, weinberger2019reliability, xu2022distributed, salehkalaibar2018distributed, espinosa2024survey}), and our technique may be applied to those settings as well.

\subsection{Notation}
We adopt the following standard notation: for any sequence $(X_1, X_2, \cdots, X_n)$ of random variables, we define
 \begin{align*}
 & X^n = (X_1, X_2, \ldots, X_n), \\
  & X^i = (X_1, \ldots, X_i), \\
 & X_{i+1}^n = (X_{i+1}, X_{i+2}, \ldots, X_n),\\
 & X_{\sim i} = (X_1, \ldots, X_{i-1}, X_{i+1}, \ldots, X_n).
 \end{align*} 
For two distributions $\hat{Q}_X$ and $Q_X$, we say that $\hat Q_X\ll Q_X$ if $\hat Q_X(x)=0$ whenever $Q_X(x)=0$.

All the logs in this paper are in base $e$, unless stated otherwise. Similarly, all the entropies are measured in nats.

\section{Upper bound using the auxiliary receiver approach}
\subsection{The auxiliary receiver approach}
Here, we review the add-and-subtract idea of the auxiliary receiver approach \cite{gon21}.
Let us view $E(P_{XY},Q_{XY},R)$ as a function of $P_{XY},Q_{XY},R$. 
The domain of this function is the set of all triples $(P_{XY},Q_{XY},R)$.  
To study the curvature of this function, we can take another point $(P_{X'Y'},Q_{X'Y'},R')$ in the domain, and study the difference $E(P_{XY},Q_{XY},R)-E(P_{X'Y'},Q_{X'Y'},R')$.  Figure \ref{firstfig}  visualizes this.  
In this work, we use the second point $(P_{X'Y'}, Q_{X'Y'}, R')$ as
$$(P_{XZ}, Q_{XZ}, R)$$
so that $X'=X, Y'=Z, R'=R$ for some "auxiliary receiver $Z$". 
The difference $E(P_{XY}, Q_{XY}, R) - E(P_{XZ}, Q_{XZ}, R)$ quantifies the curvature of the plot of the function.
We would like to find an upper bound $G_{P_{XY},Q_{XY},P_{XZ},Q_{XZ}}(R)$ on the gap $$E(P_{XY}, Q_{XY}, R) - E(P_{XZ}, Q_{XZ}, R)\leq G_{P_{XY},Q_{XY},P_{XZ},Q_{XZ}}(R).$$
Finding an upper bound on this gap leads to an upper bound on $E(P_{XY},Q_{XY},R)$ via an add-and-subtract manipulation as follows: 
 \begin{align*}
            &E(P_{XY},Q_{XY},R)\\
            &=E(P_{XY},Q_{XY},R)-E(P_{XZ},Q_{XZ},R)+E(P_{XZ},Q_{XZ},R)
            \\&\leq       G_{P_{XY},Q_{XY},P_{XZ},Q_{XZ}}(R)+E(P_{XZ},Q_{XZ},R).
            \end{align*}

The auxiliary receiver $Z$ is chosen such that either the value of $E(P_{XZ}, Q_{XZ}, R)$ is known, or an upper bound on it is available.  
The special case where $Y$ is a function of $Z$ corresponds to a Genie-aided proof. The "side information" approach of Rahman and Wagner \cite{RahmanWagner2011} is an example of such a choice (the random variable $Z$ in the bound of Rahman and Wagner corresponds to the auxiliary receiver $Z'=(Y,Z)$). However, the auxiliary receiver $Z$ does not need to be an enhanced version of $Y$, and this is our main idea to obtain new upper bounds.

\begin{figure}[htbp]% [htbp] are placement options  
\centering % Center the figure (and the tikzpicture within it) 
\begin{tikzpicture}

% Axes with rotated view
  \draw[->, thick] (0,0,0) -- (3,0,0) node[below right]{};
  \draw[->, thick] (0,0,0) -- (0,3,0) node[left]{$E$};
  \draw[->, thick] (0,0,0) -- (0,0,3) node[below left]{};

  % Origin
  % \node at (0,0,0) [left]{$O$};

  % Concave curve (parabola) in the x-y plane
  % \draw[blue, thick, domain=0.4:2, smooth, variable=\x] 
  %   plot ({\x}, {-\x*\x + 3*\x}, 2) 
  %   node[above right]{$E(P_{XY}, Q_{XY}, R)$};
  % \draw[blue, thick, domain=0.4:2.4, smooth, variable=\x] 
  %   plot ({\x}, {-\x*\x + 3*\x}, 2);

  %   \draw[blue, thick, domain = -0.2:0.4, smooth, variable=\x]
  %   plot ({\x}, {\x*\x + 1*\x+2}, 2);

  % blue part
  \draw[blue, thick, domain=-0.1:1.95, smooth, variable=\x] 
    plot ({\x}, {-1.2*(\x+0.1)*(\x-1.1)*(\x-2) + 1.5}, 1);
  \node at (2,2.3,2) [above right, text=blue]{$E(P_{XY}, Q_{XY}, R)$};
  \draw (2,2.3,2) node[circle, fill=blue, inner sep=.8pt]{};
  % dashed line to the ground
  \draw[dashed, gray] (2,0.2,2) -- (2,2.3,2);
  \draw (2,0.2,2) node[circle, fill=blue, inner sep=.8pt]{};
  \node at (2,0.2,2) [right, text=blue]{$(P_{XY}, Q_{XY}, R)$};
  
    % make it 3d
    \draw[blue, thick, domain=1.6:1.85, smooth, variable=\x] 
    plot ({\x}, {-10*(\x-1.6)*(\x-1.61)+1.9}, 1);
    \draw[blue, thick, dashed, domain=1.35:1.6, smooth, variable=\x] 
        plot ({\x}, {-10*(\x-1.6)*(\x-1.61)+1.9}, 1);
    
    \draw[blue, thick, domain=0.6:0.85, smooth, variable=\x] 
    plot ({\x}, {-10*((\x+1)-1.6)*((\x+1)-1.61)+0.9}, 1);
    
    \draw[blue, thick, dashed, domain=0.35:0.85, smooth, variable=\x] 
    plot ({\x}, {-10*((\x+1)-1.6)*((\x+1)-1.61)+0.9}, 1);

% green part
 {\node at (0.4,1.5,2) [left, text=green] {$E(P_{X'Y'},Q_{X'Y'},R')$};
 \draw (0.4,1.57,2) node[circle, fill=green, inner sep=.8pt]{};

  % dashed line to the ground
  \draw[dashed, gray] (0.4,0,2) -- (0.4,1.5,2);
  \draw (0.4,0,2) node[circle, fill=green, inner sep=.8pt]{};
  \node at (0.4,-0.1,2) [below right, text=green]{$(P_{X'Y'}, Q_{X'Y'}, R')$};}

    % mark the difference
    \draw[dashed, red] (0.4,1.57,2) -- (2,1.9,2);

    \draw[red, thick] (2,1.9,2) -- (2,2.3,2);

    \node at (2,1.7,2) [right, text=red]{$E(P_{XY}, Q_{XY}, R) - E(P_{X'Y'}, Q_{X'Y'}, R')$};

\end{tikzpicture}  
\caption{}
\label{firstfig}
\end{figure}
\subsection{Review of the add-and-subtract technique}
The add-and-subtract technique works as follows:
\begin{enumerate}
    \item Begin with a multi-letter expression that can be utilized to derive single-letter upper bounds on capacity. For example, in the point-to-point communication setting, when transmitting a message \( M \) over a channel \( p(y|x) \), we can consider \( I(M;Y^n) \) as a multi-letter proxy for capacity. In the case of a relay channel \( p(y,y_r|x,x_r) \), where \( x_r \) and \( y_r \) represent the relay's input and output, we can treat \( I(M;Y^n) \) or \( I(M;Y^nY_r^n) \) as multi-letter expressions serving as upper bounds on the communication rate.

    \item Introduce a new auxiliary receiver, denoted as \( Z \). Consider a new expression by substituting the auxiliary receiver \( Z \) for one of the original receivers. For instance, we can consider \( I(M;Z^n) \) or \( I(M;Z^nY_r^n) \) instead of \( I(M;Y^n) \) or \( I(M;Y^nY_r^n) \). We call this term the \emph{substituted term}.

    \item Add and subtract the substituted term from the original multi-letter expression. For instance, if $I(M;Y^n)$ is the original expression and $I(M;Z^n)$ is the substituted term, we can write:
    \[\underbrace{I(M;Y^n)}_{\text{original term}}=
    \underbrace{I(M;Y^n) - I(M;Z^n)}_{\text{difference term}} + \underbrace{I(M;Z^n)}_{\text{substituted term}}.
    \]

    \item Single-letterize the "difference term" using a telescopic sum expansion. For instance,  \( I(M;Y^n) - I(M;Z^n) \) may be single-letterized as follows:
    \begin{align*}
        &I(M;Y^n) - I(M;Z^n) 
        = \sum_{i} \left[ I(M;Y^{i}Z_{i+1}^n) - I(M;Y^{i-1}Z_{i}^n) \right].
    \end{align*}

    \item Lastly, reduce the "substituted term" to a single-letter form by following the procedure of an existing upper bound. This is done by mimicking the bound while formally replacing a legitimate receiver with the auxiliary receiver. For example, in the context of a relay channel, consider the cut-set bound and substitute the original receiver $Y$ with the auxiliary receiver $Z$ throughout all the steps of the bound.
\end{enumerate}
This method produces novel upper bounds for relay, broadcast, and interference channels \cite{gon21}. 

A formal proof of applying the add-and-subtract technique to the distributed hypothesis testing problem is presented in Section \ref{sec:addsubtractdetails}. In the remainder of this section, we provide an intuitive discussion of the concept.

For the distributed hypothesis testing problem, we consider $D(P_{MY^n} \| Q_{MY^n})$ as our multi-letter expression (see Theorem \ref{thm1cz}). After replacing $Y$ with an auxiliary receiver $Z$, we obtain the expression $D(P_{MZ^n} \| Q_{MZ^n})$
as our "substituted term". Finally, we form and expand the difference term
$D(P_{MY^n} \| Q_{MY^n})-D(P_{MZ^n} \| Q_{MZ^n})$
via a telescopic sum as follows:
\begin{align}
&\frac1n{\color{black}D(P_{M Y^n} \| Q_{M Y^n} ) - \frac1n D(P_{M Z^n} \| Q_{M Z^n} )} \nonumber
\\ &= \frac1n\sum_{i=1}^n \left [  D(P_{M Y_{1:i-1} Y_i Z_{i+1:n}} \| Q_{M Y_{1:i-1} Y_i Z_{i+1:n}})  - D(P_{M Y_{1:i-1} Z_i Z_{i+1:n}} \| Q_{M Y_{1:i-1} Z_i Z_{i+1:n}}) \right ] \nonumber
\\ &= \frac1n\sum_{i=1}^n \left [  D(P_{U_iY_i} \| Q_{U_i Y_i}) - D(P_{U_iZ_i} \| Q_{U_iZ_i}) \right ] \nonumber
\end{align} 
where $U_i=M Y_{1:i-1}Z_{i+1:n}$. With this identification for $U_i$, one can also show that
$\sum I(U_i;X_i|Y_i)\leq nR$ and $\sum I(U_i;X_i|Z_i)\leq nR$. 
Using the above equations, we can show that the difference term is less than or equal to 
$$\sup_{P_{U|X},Q_{U|X}} [D(P_{UY}\|Q_{UY})-D(P_{UZ}\|Q_{UZ})]$$
where $P_{U|X}$ and $Q_{U|X}$ satisfy $I_P(U;X|Y)\leq R$, $I_P(U;X|Z)\leq R$, $I_Q(U;X|Y)\leq R$ and $I_Q(U;X|Z)\leq R$. We further simplify this expression by dropping the constraints \(I_Q(U;X|Y) \leq R\) and \(I_Q(U;X|Z) \leq R\), and by bounding the objective function from above. This allows us to obtain a simpler expression for comparison with previously known bounds. In  Section \ref{sec:addsubtractdetails},  we define $G_{P_{XY},Q_{XY},P_{XZ},Q_{XZ}}(R)$ as an upper bound on the difference term. 
 This completes Step 4 of the above procedure and yields:
 \begin{align*}
 &\underbrace{\frac1n D(P_{M Y^n} \| Q_{M Y^n})}_{\text{original term}} \\
 &=
    \underbrace{\frac1n{\color{black}D(P_{M Y^n} \| Q_{M Y^n} ) - \frac1n D(P_{M Z^n} \| Q_{M Z^n} )}}_{\text{difference term}} + \underbrace{\frac1n D(P_{M Z^n} \| Q_{M Z^n} )}_{\text{substituted term}}
    \\&\leq G_{P_{XY},Q_{XY},P_{XZ},Q_{XZ}}(R)+\underbrace{\frac1n D(P_{M Z^n} \| Q_{M Z^n} )}_{\text{substituted term}}.
    \end{align*}
 We do not perform Step 5 in this paper and will state the bound simply as follows:
\begin{align*}
 & E_{P_{XY},Q_{XY}}(R) \leq G_{P_{XY},Q_{XY},P_{XZ},Q_{XZ}}(R)+E_{P_{XZ},Q_{XZ}}(R).
\end{align*}
Indeed, any upper bound on $E_{P_{XZ},Q_{XZ}}(R)$ can be used to derive an upper bound on $E_{P_{XY},Q_{XY}}(R)$. In particular, using the centralized bound $E_{P_{XZ},Q_{XZ}}(R)\leq D(P_{XZ}\|Q_{XZ})$, we obtain
\begin{align*}
 & E_{P_{XY},Q_{XY}}(R) \leq G_{P_{XY},Q_{XY},P_{XZ},Q_{XZ}}(R)+D(P_{XZ}\|Q_{XZ}).
\end{align*}
However, it is unclear whether using the centralized upper bound is the best choice.

\begin{remark}
   It is useful to compare the distributed hypothesis testing problem with the source model problem in the key agreement context (where the choice of upper bound in Step 5 is clear). We refer the reader to Chapter 22 of \cite{elk11} for the definition of the key rate capacity \( S(X;Y\|Z) \). Given a source \( P_{XYZ} \) and any auxiliary receiver \( J \), we can show that:
\[
S(X;Y\|Z) \leq \Delta(P_{XYZ}, P_{XYJ}) + S(X;Y\|J),
\]
where 
\[
\Delta(P_{XYZ}, P_{XYJ}) = \max_{P_{UV|XY}} [I(U;Z|V) - I(U;J|V)].
\]
The term \( \Delta(P_{XYZ}, P_{XYJ}) \) serves as an upper bound on the "difference term", and can be obtained through a telescopic sum expansion. In the case of the key agreement problem, we can complete Step 5. Notably, if we adopt the simple upper bound \( S(X;Y\|J) \leq I(X;Y|J) \), we arrive at the following upper bound for \( S(X;Y|Z) \), which is currently the best known upper bound for this problem \cite{gohari2010information}:
\[
S(X;Y\|Z) \leq \Delta(P_{XYZ}, P_{XYJ}) + I(X;Y|J).
\]
\end{remark}

\subsection{The bound obtained using the add-and-subtract technique}
\label{sec:addsubtractdetails}

\begin{definition}
    Let $\mathcal{T}(P_{XY},P_{XZ})$ to be the set of all $P_{U|X}$ such that 
    \begin{align*}
    R\geq \max[I_P(U;X|Z), I_P(U;X|Y)]
\end{align*}
where the mutual information terms $I_P(U;X|Z), I_P(U;X|Y)$ are computed according to $P_{U|X}P_{XZ}$ and $P_{U|X}P_{XY}$ respectively.
\end{definition}

\begin{definition}
    For any $P_X$, $Q_X$, $P_{Y|X}$, $P_{Z|X}$, $Q_{Y|X}$ and $Q_{Z|X}$, we define
\begin{align*}
 & f_{P_{Y|X}, P_{Z|X}, Q_{Y|X}, Q_{Z|X}}(P_X) = \max_{\hat{Q}_X } [D(P_Y \| \hat{Q}_Y) - D(P_Z \| \hat{Q}_Z) ]
\end{align*} 
where the maximum is over all $\hat{Q}_X\ll Q_X$ and
$$\hat{Q}_Y(y)=\sum_x\hat{Q}_X(x)Q_{Y|X}(y|x),$$
$$\hat{Q}_Z(z)=\sum_x\hat{Q}_X(x)Q_{Z|X}(z|x).$$
Next, let 
\begin{align}
    &G_{P_{XY},Q_{XY},P_{XZ},Q_{XZ}}(R) =\max_{P_{U|X}\in\mathcal{T}(P_{XY},P_{XZ})} \sum_u P(u) f( P_{X|U}(\cdot | u)).\label{eqnDefCf}
\end{align}
\label{f_definition}
\end{definition}
\begin{remark}
Using the standard application of the Caratheodory theorem, in equation \eqref{eqnDefCf}, it suffices to restrict to $|\mathcal U|\leq |\mathcal X|+2$. 
\end{remark}

\begin{theorem}\label{theorem_1} 
Assume that $P_{XYZ}$ and $Q_{XYZ}$ satisfy $D(P_{XY}\|Q_{XY})<\infty$ and $D(P_{XZ}\|Q_{XZ})<\infty$. Then, we have
\begin{align*}
 & E_{P_{XY},Q_{XY}}(R) \leq G_{P_{XY},Q_{XY},P_{XZ},Q_{XZ}}(R)+E_{P_{XZ},Q_{XZ}}(R).
\end{align*}
\label{main_Thm}
\end{theorem}
\iffalse
\begin{corollary}
    If $P_{XY}$ and $Q_{XY}$ and $P_{XZ}$ and $Q_{XZ}$ satisfy $D(P_{XY}\|Q_{XY})<\infty$ and $D(P_{XZ}\|Q_{XZ})<\infty$, we have
\begin{align*}
 & E_{P_{XY},Q_{XY}}(R) \leq \mathcal{C} f({P}_X)+E_{P_{XZ},Q_{XZ}}(R)
\end{align*}
where $\mathcal{C} f({P}_X)$
is the upper concave envelope of $f_{P_{Y|X}, P_{Z|X}, Q_{Y|X}, Q_{Z|X}}(\cdot)$ at $P_X$. 
\end{corollary}
\fi
\begin{proof}

For a given $n$, we assume that $(X^n,Y^n,Z^n)$ are i.i.d.\ under the two hypothesis:
$$P_{X^n,Y^n,Z^n}=\prod_{i=1}^nP_{X_i,Y_i,Z_i},$$
$$Q_{X^n,Y^n,Z^n}=\prod_{i=1}^nQ_{X_i,Y_i,Z_i}.$$
Given some encoding map from $X^n$ to $M$, we assume
$$P_{M,X^n,Y^n,Z^n}=P_{M|X^n}P_{X^n,Y^n,Z^n},$$
$$Q_{M,X^n,Y^n,Z^n}=Q_{M|X^n}Q_{X^n,Y^n,Z^n}.$$
We have $P_{M|X^n}=Q_{M|X^n}$ as the encoding map \(f_n(\cdot)\) is the same under the two hypotheses.

Take some $\epsilon>0$. Let $n$ and $P_{M|X^n}$ be 
such that
\[E_{P_{XY},Q_{XY}}(R)\leq \frac{1}{n} D(P_{MY^n} \| Q_{MY^n})+\epsilon.\]

Next, we utilize the add-and-subtract technique:
\begin{align*}
 & \frac1n D(P_{M Y^n} \| Q_{M Y^n} )  
 \\&= {\color{black}\frac1n D(P_{M Y^n} \| Q_{M Y^n} ) - \frac1n D(P_{M Z^n} \| Q_{M Z^n} )} + {\color{black}\frac1n D(P_{M Z^n} \| Q_{M Z^n} ) }.
\end{align*} 
From Theorem \ref{thm14Ahlswede}, we have
$${\color{black}\frac1n D(P_{M Z^n} \| Q_{M Z^n} ) }\leq E_{P_{XZ},Q_{XZ}}(R).$$
Next, we use the telescopic-sum idea of the add-and-subtract technique to expand the difference term:
\begin{align}
&{\color{black}D(P_{M Y^n} \| Q_{M Y^n} ) - D(P_{M Z^n} \| Q_{M Z^n} )} \nonumber
\\ &= \sum_{i=1}^n [ D(P_{M Y_{1:i-1} Y_i Z_{i+1:n}} \| Q_{M Y_{1:i-1} Y_i Z_{i+1:n}}) - D(P_{M Y_{1:i-1} Z_i Z_{i+1:n}} \| Q_{M Y_{1:i-1} Z_i Z_{i+1:n}}) ]\label{eqnf0}
\end{align} 
where $D(P_{MY^{i}Z_{i+1}^n}\|Q_{MY^{i}Z_{i+1}^n})$ is finite for every $i$ because 
\begin{align*}
    & D(P_{MY^{i}Z_{i+1}^n}\|Q_{MY^{i}Z_{i+1}^n}) 
    \\ &\leq D(P_{X^nY^{i}Z_{i+1}^n}\|Q_{X^nY^{i}Z_{i+1}^n}) 
    \\&= 
iD(P_{XY}\|Q_{XY})+(n-i)D(P_{XZ}\|Q_{XZ}) 
\\&<\infty.
\end{align*}

Let $$U_i = (M,Y_{1 : i-1},Z_{i+1 : n}).$$ 
Let $T\in\{1,2,\cdots, n\}$ be a time-sharing random variable, independent of all previously defined variables, and set
$$U=(T,U_T), \quad X=X_T,\quad Y=Y_T,\quad Z=Z_T.$$

Observe that
\begin{align*}
 & \sum_{i=1}^n [  D(P_{M Y_{1:i-1} Y_i Z_{i+1:n}} \| Q_{M Y_{1:i-1} Y_i Z_{i+1:n}}) 
 - D(P_{M Y_{1:i-1} Z_i Z_{i+1:n}} \| Q_{M Y_{1:i-1} Z_i Z_{i+1:n}}) ] \\
 &= \sum_{i=1}^n [D(P_{U_i Y_i} \| Q_{U_i Y_i}) - D(P_{U_i Z_i} \| Q_{U_i Z_i})] \\
 &= \sum_{i=1}^n [\underbrace{D(P_{U_i} \| Q_{U_i})}_{\text{cancelled}} + D(P_{Y_i | U_i} \| Q_{Y_i|U_i} | P_{U_i})  
 - \underbrace{D(P_{U_i} \| Q_{U_i})}_{\text{cancelled}} - D(P_{Z_i | U_i} \| Q_{Z_i | U_i} | P_{U_i})].
\end{align*} 
Observe that under either $P_{M,X^n,Y^n,Z^n}$ or $Q_{M,X^n,Y^n,Z^n}$, we have the following Markov chain:
\[
U_i \rightarrow X_i \rightarrow Y_i Z_i
\]
as
\(
0 \leq I(M Y_{1:i-1} Z_{i+1:n} ; Y_i Z_i | X_i) \leq I(X^n Y_{1:i-1} Z_{i+1:n} ; Y_i Z_i | X_i) = 0
\) by the i.i.d assumption. Then
\begin{align*} 
&\sum_{i=1}^n  [D(P_{Y_i | U_i} \| Q_{Y_i\|U_i} | P_{U_i})  - D(P_{Z_i | U_i} \| Q_{Z_i | U_i} | P_{U_i})]
\\
 &=\sum_{i=1}^n \sum_{u_i} P_{U_i}(u_i) \left[  D(P_{Y_i|u_i} \| Q_{Y_i|u_i} ) - D( P_{Z_i|u_i} \| Q_{Z_i|u_i} )  \right] 
 \\
 &=n\cdot \sum_{u} P_{U}(u) \left[  D(P_{Y|u} \| Q_{Y|u} ) - D( P_{Z|u} \| Q_{Z|u} )  \right] 
 \\
 &\leq n\cdot \sum_{u} P_{U}(u) \max_{\hat{Q}_X \ll Q_X}\left[  D(P_{Y|u} \| \hat Q_{Y} ) - D( P_{Z|u} \| \hat Q_{Z} )  \right]
\end{align*} 
where the maximum is over $\hat{Q}_X\ll Q_X$ (observe that if $Q_X(x)=0$, then $Q_{X|u}(x)=0$) and $\hat Q_Y$, $\hat Q_Z$ are given by Definition \ref{f_definition}. 

Next, we have $I_P(U ; X | Z)\leq R$ since 
 \begin{align}
  nR&\geq H_P(M)\label{eqnf2}
  \\&\geq I_P(M; X^n | Z^n)\label{eqnf3}
 \\&= \sum_{i=1}^{n} I_P(M; X_i | Z^n X^{i-1}) \nonumber\\&=\sum_{i=1}^{n}I_P(M X^{i-1} Z_{\sim i} ; X_i | Z_i) \label{eqnMXZ1}
\\&=\sum_{i=1}^{n}I_P(M X^{i-1}Y^{i-1} Z_{\sim i} ; X_i | Z_i) \label{eqnMXZ2}
\\&\geq \sum_{i=1}^{n}I_P(M Y^{i-1} Z_{i+1}^n ; X_i | Z_i)\nonumber
\\&=\sum_{i=1}^{n}I_P(U_i ; X_i | Z_i) \nonumber
\\&=n\cdot I_P(U ; X | Z) \nonumber
 \end{align}
where \eqref{eqnMXZ1} follows from the i.i.d. assumption, and \eqref{eqnMXZ2} follows from \(I(Y^{i-1}; X_i | Z^n X^{i-1} M) = 0\) as
 \begin{align*}
 I(Y^{i-1}; X_i | Z^n X^{i-1} M) &\leq I(Y^{i-1}; X_iM | Z^n X^{i-1}) \\
 &\leq I(Y^{i-1} ; X^n | Z^n X^{i-1}) \\
 &= I(Y^{i-1}; X^i | Z^n ) = 0.
 \end{align*} 
The proof for $
    R\geq I_P(U ; X | Y)$ 
is similar.
This completes the proof.
\end{proof}

The following theorem yields an explicit upper bound in a special case, which is useful in comparing our bound with the previously known bounds:
\begin{theorem}\label{theorem_wagner}
    Take an arbitrary pair \((P_{Z|XY}, Q_{Z|XY})\) such that \(
    Q_{YZ|X}= Q_{Z|X} Q_{Y|Z}
    \). Then for any $P_{UX}$ we have
    \begin{align}
    \sum_u P(u) f( P_{X|U}(\cdot | u)) \leq  D(P_{YZU} \| P_{UZ} Q_{Y|Z}). \label{rhsthm5}
    \end{align}
\end{theorem}

\begin{proof}
Given some $\hat{Q}_X$, let $$\hat{Q}_{XYZ}=\hat Q_X Q_{YZ|X}=\hat Q_XQ_{Z|X}Q_{Y|Z}
.$$ Then, we have
\begin{align}
&\sum_u P(u)\max_{\hat{Q}_X} [D(P_{Y|u} \| \hat{Q}_Y) - D(P_{Z|u} \| \hat{Q}_Z)]\nonumber \\&\leq \sum_u P(u)\max_{\hat{Q}_X} [D(P_{YZ|u} \| \hat{Q}_{YZ}) - D(P_{Z|u} \| \hat{Q}_Z)]\label{eqnTL1}
\\&=\sum_u P(u)\max_{\hat{Q}_X} \sum_z P(z|u) D(P_{Y|z,u} \| \hat{Q}_{Y|z})\nonumber
\\&=\sum_u P(u)\max_{\hat{Q}_X} \sum_z P(z|u) D(P_{Y|z,u} \| Q_{Y|z})\label{eqnTL2}
\\&=\sum_u P(u)\sum_z P(z|u) D(P_{Y|z,u} \| Q_{Y|z})\nonumber
\\&=D(P_{YZU} \| P_{UZ} Q_{Y|Z}) \nonumber
\end{align}
where \eqref{eqnTL1}
follows from the data processing inequality, and 
\eqref{eqnTL2} follows from the definition of \(\hat{Q}_{XYZ}\).  
\end{proof}

\begin{remark}
    {Please note that, unlike the left-hand side, the right-hand side in \eqref{rhsthm5} does not depend on $Q_{XZ}$, so one may assume $Q_{XZ}=P_{XZ}$ when using Theorem \ref{theorem_wagner}.}
\end{remark}
\section{Comparison with previous bounds}
Next, we show that the bound in Theorem \ref{theorem_1} improves over the bounds given in \cite{Yuval2024} and \cite{RahmanWagner2011}.

\begin{lemma}
    \label{lemfG}

Given any \((P_{Z|XY}, Q_{Z|XY}) \in \mathcal{R}(P_{XY}, Q_{XY})\) as defined in Definition \ref{definition_1} and for every $P_{X|U}$, we have
\begin{align}
    \sum_u P(u)f(P_{X|U}(\cdot | u)) \leq I_P(Y; U | Z) + D(P_{YZ} \| Q_{YZ}). \label{eqn_wagner_compare}
\end{align}
Consequently,
\begin{align*}
    &G_{P_{XY},Q_{XY},P_{XZ},Q_{XZ}}(R) 
    \leq \max_{P_{U|X}\in\mathcal{T}(P_{XY},P_{XZ})} I_P(Y; U | Z) + D(P_{YZ} \| Q_{YZ}).
\end{align*}
\end{lemma}

\begin{proof}
From the definition of \( \mathcal{R}(P_{XY}, Q_{XY})\) we obtain 
% \begin{align*}
%     P_{UZ} = Q_{UZ},\quad & Q_{UZY} = Q_{UZ} Q_{Y|Z},\quad Q_{U|YZ} = Q_{U|Z}.
% \end{align*}
\begin{align*}
    P_Z = Q_Z.
\end{align*}
Then immediately,
\begin{align*}
    D(P_{YZU} \| P_{UZ} Q_{Y|Z}) 
    &= \sum_{u,y,z} P_{YZU} \log \Big ( \frac{P_{YZU} Q_Z}{P_{UZ} Q_{YZ}} \Big) 
    \\&= \sum_{u,y,z} P_{YZU} \log \Big ( \frac{P_{YU|Z}}{P_{U|Z}P_{Y|Z}} \cdot \frac{P_{YZ}}{Q_{YZ}} \Big )
    \\& = I_P(Y;U |Z) + D(P_{YZ} \| Q_{YZ}).
\end{align*}

Applying Theorem \ref{theorem_wagner} completes the proof.

\end{proof}
\begin{corollary}\label{corr-weaker}
    Given that \((P_{Z|XY}, Q_{Z|XY}) \in \mathcal{R}(P_{XY}, Q_{XY})\), we have
$$E_{P_{XY},Q_{XY}}(R) \leq \max_{P_{U|X}\in\mathcal{T}(P_{XY},P_{XZ})} I_P(Y; U | Z) + D(P_{YZ} \| Q_{YZ}). $$
This corollary follows from $E_{P_{XZ},Q_{XZ}}(R)=0$ by the assumption $P_{XZ}=Q_{XZ}$ from $\mathcal{R}(P_{XY}, Q_{XY})$.
\end{corollary}

Note that the bound in Corollary \ref{corr-weaker} is a weaker version of the auxiliary receiver bound. The bound in Corollary \ref{corr-weaker} is less than or equal to the bound in \cite{RahmanWagner2011} since $\mathcal{T}(P_{XY},P_{XZ})\subset \mathcal{S}_{P_{XYZ}}$ as $\mathcal{S}_{P_{XYZ}}$ is the set of $P_{U|X}$ satisfying
$$R\geq I_P(X;U|Z)$$
while $P_{U|X}\in\mathcal{T}(P_{XY},P_{XZ})$ satisfies 
$$R\geq \max(I_P(X;U|Y),I_P(X;U|Z)).$$
Next, as a sanity check, we consider the example of "test against conditional independence" from \cite{RahmanWagner2011}:
\begin{example}[Test against conditional independence]
Assume that $Y=(J,Z,\hat Y)$ for some random variables $J$, $Z$ and $\hat Y$. We denote the distribution of $J,Z,\hat Y$ under distribtion $P$ and $Q$ by $P_{J,Z,\hat Y}$ and $Q_{J,Z,\hat Y}$. Moreover, assume that $P_{XY}=P_{XJZ\hat Y}$ and $Q_{XY}=Q_{XJZ\hat Y}$  satisfy the following constraints:    
    $I_{Q}(\hat Y;XJ|Z)=0$, $P_{XJZ}=Q_{XJZ}$ and $P_{\hat YZ}=Q_{\hat YZ}$. Then, the upper bound in Corollary \ref{corr-weaker} implies the following bound:
    \begin{align*}
 & E_{P_{XY},Q_{XY}}(R)
 \leq \max_{P_{U|X}}
I_P(\hat Y; UJ | Z)
\end{align*}
where the maximum is over $P_{U|X}$ satisfying $I_P(U;X|ZJ)\leq R$.\label{corrJ}
\end{example}

\begin{proof}[Proof of Example \ref{corrJ}] Let $Z'=(Z,J)$. Note that \eqref{eqnDef2an} and \eqref{eqnDef2bn} hold for $(X,Y,Z')=(X,Y,ZJ)$, i.e.,
\begin{align}
      Q_{YZJ|X} &= Q_{ZJ|X} Q_{Y|ZJ}, \label{eqnDef2anN} \\
     P_{XZJ} &= Q_{XZJ}, \label{eqnDef2bnN} 
\end{align} 
due to our assumptions. Then, \((P_{Z'|XY}, Q_{Z'|XY}) \in \mathcal{R}(P_{XY}, Q_{XY})\) and we obtain
\begin{align*}
&E_{P_{XY},Q_{XY}}(R) 
\\&\leq \max_{P_{U|X}\in\mathcal{T}(P_{XY},P_{XZ'})} I_P(Y; U | Z') + D(P_{YZ'} \| Q_{YZ'})
\\&=\max_{P_{U|X}} I_P(Y; U | ZJ) + D(P_{YZJ} \| Q_{YZJ})
    \\&=\max_{P_{U|X}} I_P(Y; U | ZJ) + D(P_{Y} \| Q_{Y})
\end{align*}
where the maximum is over $P_{U|X}$ satisfying $I_P(U;X|YJ)\leq R$ and $I_P(U;X|ZJ)\leq R$.
Observe that
$I_P(U;X|YJ)\leq I_P(U;X|ZJ)$ because $Z$ is a function of $Y$. Thus, $I_P(U;X|YJ)\leq R$ is redundant. Further, we have
\begin{align*}
    &I_P(Y; U | ZJ) + D(P_{Y} \| Q_{Y})
    \\&=I_P(\hat Y; U | ZJ) 
    + D(P_{\hat YZJ} \| Q_{\hat YZJ})
    \\&=I_P(\hat Y; U | ZJ) 
    + D(P_{\hat YZJ} \| P_{\hat YZ}Q_{J|\hat YZ})
    \\&=I_P(\hat Y; U | ZJ) 
    + D(P_{\hat YZJ} \| P_{\hat YZ}Q_{J|Z})
    \\&=I_P(\hat Y; U | ZJ) 
    + D(P_{\hat YZJ} \| P_{\hat YZ}P_{J|Z})
     \\&=I_P(\hat Y; U | ZJ) + I_P(\hat Y;J|Z)
          \\&=I_P(\hat Y; UJ | Z).
\end{align*}
\end{proof}

\subsection{The minimum communication rate needed to achieve the centralized bound }
The centralized bound states that
\begin{align}
E_{P_{XY},Q_{XY}}(R)\leq D(P_{XY}\|Q_{XY}). \label{centralized_bound2}
\end{align}
The centralized bound is tight for large $R$. But exactly how large does $R$ need to be? Let $R^*$ be the smallest rate such that
\begin{align}
E_{P_{XY},Q_{XY}}(R)= D(P_{XY}\|Q_{XY}), \qquad \forall R\geq R^*. 
\end{align}
It is shown in  \cite{Shimokawa1994} that  $R^*\leq D(P_{XY}\|Q_{XY})+H_P(X|Y)$. In this section, we use our new bound to show that if $P_{XY}$ and $Q_{XY}$ are non-degenerate, we have $R^*\geq H_P(X|Y)$. 

\begin{definition}
    Let $\mathcal V$ be the class of $(P_{Y|X}, Q_{Y|X})$ such that 
    \begin{itemize}
        \item $P_{Y|X}(y | x)>0$ and $Q_{Y|X}(y | x)>0$ for all $x \in \mathcal X,y \in \mathcal Y$.
        \item  One \textbf{cannot} find $x_0, x_1 \in \mathcal X, x_0 \neq x_1$ such that 
    $$
    \frac{P_{Y|X}(y_0 | x_0) P_{Y|X}(y_1 | x_1)}{P_{Y|X}(y_0|x_1)P_{Y|X}(y_1|x_0)} = \frac{Q_{Y|X}(y_0 | x_0) Q_{Y|X}(y_1 | x_1)}{Q_{Y|X}(y_0|x_1)Q_{Y|X}(y_1|x_0)},\qquad\qquad \forall y_0\neq y_1 \in \mathcal Y.
    $$
    \end{itemize}
   \label{defin4}
\end{definition}
Observe that unless $P_{Y|X}$ and $Q_{Y|X}$ are structured, they will belong to $\mathcal{V}$, i.e., if $P_{Y|X}$ and $Q_{Y|X}$ are uniformly chosen at random from the class of all channels, with probability one, the pair will belong to $\mathcal{V}$. 

\begin{theorem}
   For any $(P_{Y|X}, Q_{Y|X}) \in \mathcal V$, $R^*\geq H_P(X|Y)$.\label{thm-general-position} 
\end{theorem}

\begin{proof}
 Take an arbitrary $R < H_P(X|Y)$. It suffices to show that the upper bound in Theorem \ref{main_Thm} is strictly smaller than the centralized bound.
    Assume that $Z = X$ under both $P$ and $Q$. Then,
    $E_{P_{XZ},Q_{XZ}}(R)=D(P_X\|Q_X)$ and the bound in Theorem \ref{main_Thm} reduces to the following:
\begin{align}
    E_{P_{XY},Q_{XY}}(R)\leq D(P_X\|Q_X)+\max_{P_{U|X}:I_P(U;X|Y)\leq R}\sum_{u}P_U(u) \max_{\hat Q_X} \left \{ D[P_{Y|u} \| \hat Q_Y] - D[P_{X|u} \| \hat Q_X] \right \}. \label{rhs19eq}
    \end{align}
    The inner maximizer over $\hat{Q}_X$ depends on $u$. Let us denote the maximizer by $\hat{Q}_{X|u}$ for $u\in\mathcal{U}$. Let
    $\hat Q_{XY|u} = \hat Q_{X|u}Q_{Y|X}$.
    Observe that
  \begin{align}  
    &\sum_{u }P_U(u)  \left \{ D[P_{Y|u} \| \hat Q_{Y|u}] - D[P_{X|u} \| \hat Q_{X|u}] \right \} \nonumber
    \\& \leq \sum_{u}P_U(u)  \left \{ D[P_{XY|u} \| \hat Q_{XY|u}] - D[P_{X|u} \| \hat Q_{X|u}] \right \}  \label{less_cen_eq_cond} \\
    & = \sum_{u}P_U(u)  \sum_{x} P_{X|U}(x|u) D(P_{Y|Xu} \| \hat Q_{Y|Xu}) 
    \\
    & = \sum_{u}P_U(u)  \sum_{x} P_{X|U}(x|u) D(P_{Y|X} \| Q_{Y|X}) 
    \\
    & = D(P_{XY} \| Q_{XY})-D(P_X\|Q_X).
\end{align}
This shows that the upper bound in \eqref{rhs19eq} is always less than or equal to the centralized bound. 
Observe that the inequality in \eqref{less_cen_eq_cond} holds with equality only when
$$
\sum_{u, y }p(u,y) D(P_{X|Y,u} \| \hat Q_{X|Y,u}) = 0.
$$
That is, for any $u,y$ where $p(u,y) > 0$, we must have 
\begin{align}
    P_{X|Y,u} =\hat Q_{X|Y,u} \label{less_con_eq_2}.
\end{align}
Since $I_P(U;X|Y)\leq R < H_P(X|Y)$, we obtain $H_P(X|Y,U)>0$. Thus, $H_P(X|U)\geq H_P(X|Y,U)>0$.
Therefore, one can find some $u^*\in\mathcal{U}$, $x_0, x_1 \in \mathcal X, x_0 \neq x_1$ such that $P_U(u^*)>0$, $P_{X|U}(x_0|u^*) > 0, P_{X|U}(x_1|u^*) > 0$. Since $P_{Y|X}(y|x)>0$ and $Q_{Y|X}(y|x)>0$ for all $x,y$, equation \eqref{less_con_eq_2} implies
\begin{align}
    \frac{P_{X|Y,U}(x_0|y,u^*)}{P_{X|Y,U}(x_1|y,u^*)} = \frac{\hat Q_{X|Y,U}(x_0|y,u^*)}{\hat Q_{X|Y,U}(x_1|y,u^*)}, \quad \forall y.
\end{align}
Using the Markov chain $U \rightarrow X \rightarrow Y$, under both $P$ and $\hat Q$, we deduce that for all $y$ we have
\begin{align}
    \frac{P_{X|U}(x_0|u^*)P_{Y|X}(y|x_0)}{P_{X|U}(x_1|u^*)P_{Y|X}(y|x_1)} &= \frac{\hat Q_{X|U}(x_0|u^*)\hat Q_{Y|X}(y|x_0)}{\hat Q_{X|U}(x_1|u^*)\hat Q_{Y|X}(y|x_1)}\\
     &= \frac{\hat Q_{X|U}(x_0|u^*) Q_{Y|X}(y|x_0)}{\hat Q_{X|U}(x_1|u^*) Q_{Y|X}(y|x_1)}.
\end{align}
If we write the above equation for $y_0$ and $y_1$ and take their ratios, we deduce that
$$
    \frac{P_{Y|X}(y_0 | x_0) P_{Y|X}(y_1 | x_1)}{P_{Y|X}(y_0|x_1)P_{Y|X}(y_1|x_0)} = \frac{Q_{Y|X}(y_0 | x_0) Q_{Y|X}(y_1 | x_1)}{Q_{Y|X}(y_0|x_1)Q_{Y|X}(y_1|x_0)},\qquad\qquad \forall y_0\neq y_1 \in \mathcal Y.
    $$
But this contradicts  $(P_{Y|X}, Q_{Y|X}) \in \mathcal V$. This completes the proof.
\end{proof}

\begin{remark}
The conditional entropy $H_P(X|Y)$     is just a convenient bound obtained using Theorem \ref{main_Thm} with the choice of $Z=X$ under both hypotheses. In some cases, it is possible to obtain better bounds than $R^*\geq H_P(X|Y)$. See the discussion of Figure \ref{fig_against_rate} in the next section. 
\end{remark}

\subsection{Strict improvement for binary sources}
This section\footnote{MATLAB simulation codes for this and the following section can be found at \href{https://github.com/zhenduowen/DistributedHypothesisTest\_Simulations}{github.com/zhenduowen/DistributedHypothesisTest\_Simulations.} } compares our bound in Theorem \ref{theorem_1} with Rahman-Wagner's upper bound \cite{RahmanWagner2011} for doubly symmetric binary sources (DSBS). In this section for DSBS, we measure the communication rate by bits per symbol, the logarithms are in base $2$, and the exponent is also considered in base $2$.

Consider $X$ and $Y$ to be a DSBS, where under $H_0$, $(X,Y)$ is distributed as follows:

$$P_{XY}: \quad 
\begin{array}{c|cc}
X/ Y & 0 & 1 \\ \hline
0 & \frac12 (1 - \kappa_0) & \frac12 \kappa_0\\
1 & \frac12 \kappa_0 & \frac12 (1-\kappa_0)
\end{array}
$$
and under $H_1$, $(X,Y)$ is distributed as follows:
$$Q_{XY}: \quad 
\begin{array}{c|cc}
X/ Y & 0 & 1 \\ \hline
0 & \frac12 (1 - \kappa_1) & \frac12 \kappa_1\\
1 & \frac12 \kappa_1 & \frac12 (1-\kappa_1)
\end{array}
$$
for some $\kappa_0,\kappa_1\in[0,1]$. 

Note that $P(X=1) = Q(X=1) = \frac12$, and the channels from $X$ to $Y$ are binary symmetric channels (BSC): $P_{Y|X} \sim BSC(\kappa_0)$, $Q_{Y|X} \sim BSC(\kappa_1)$. 

Let
$$\rho_i=1-2\kappa_i, \qquad i=0,1.$$
Then, $\rho_0,\rho_1\in[-1,1]$ and the Pearson correlation coefficients between $X$ and $Y$ under hypothesis $H_0$ and $H_1$ are $\rho_0$ and $\rho_1$ respectively. Since the transformation $Y'=1-Y$ converts a problem with $(\rho_0,\rho_1)$ into one with $(-\rho_0,-\rho_1)$, without loss of generality, we take $\rho_0 \in [-1,1]$ and $\rho_1 \in [0,1]$. Therefore, $\kappa_1\in[0,\frac12]$.

The centralized upper bound for these sources evaluates to
$$E_{P_{XY}, Q_{XY}}(R) \leq  (1 - \kappa_0) \log_2 \left( \frac{1 - \kappa_0}{1 - \kappa_1} \right) + \kappa_0 \log_2 \left( \frac{\kappa_0}{\kappa_1} \right).$$
The centralized bound is tight \cite{Shimokawa1994} when $$R\geq D(P_{XY}\|Q_{XY})+H_P(X|Y)=(1 - \kappa_0) \log_2 \left( \frac{1}{1 - \kappa_1} \right) + \kappa_0 \log_2 \left( \frac{1}{\kappa_1} \right). $$
Moreover, Theorem \ref{thm-general-position}  shows that the centralized bound is not tight when
$$R<H_P(X|Y)=(1 - \kappa_0) \log_2 \left( \frac{1}{1 - \kappa_0} \right) + \kappa_0 \log_2 \left( \frac{1}{\kappa_0} \right).$$
Note that the pair $(P_{Y|X},Q_{Y|X})$ belongs to $\mathcal{V}$ (defined in Definition \ref{defin4}) for every $\kappa_0, \kappa_1\in(0,1)$ where $\kappa_0\neq \kappa_1$ since 
\begin{align*}
    \frac{P_{Y|X}(0 | 0) P_{Y|X}(1 | 1)}{P_{Y|X}(0|1)P_{Y|X}(1|0)} = \left(\frac{1-\kappa_0}{\kappa_0}\right)^2,\quad  \frac{Q_{Y|X}(0 | 0) Q_{Y|X}(1 | 1)}{Q_{Y|X}(0|1)Q_{Y|X}(1|0)}= \left(\frac{1-\kappa_1}{\kappa_1}\right)^2.
\end{align*}
    
\subsubsection{Rahman and Wagner's Bound}
For evaluation of Rahman and Wagner's Bound, we
need to choose 
 \((P_{Z|XY}, Q_{Z|XY})\in  \mathcal{R}(P_{XY}, Q_{XY})\) satisfying
\begin{align}
      Q_{YZ|X} &= Q_{Z|X} Q_{Y|Z}, \label{eqnDef2anD2} \\
     P_{XZ} &= Q_{XZ}.\label{eqnDef2bnD2} 
\end{align} 
To simplify our evaluation, we assume that the channel $Q_{Z|X}$ is a BSC with some parameter $\kappa'_1$. We further assume that $\kappa'_1\leq \kappa_1\leq \frac12$. This ensures that under the hypothesis $Q$,  the joint distribution $Q_{XYZ}$ of the form $Q_{XYZ} = Q_XQ_{Z|X}Q_{Y|Z}$ can be defined where the channel $Q_{Y|Z}$ would be also a BSC with cross-over probability $\frac{\kappa_1 - \kappa_1'}{1 - 2\kappa_1'}$. From \eqref{eqnDef2bnD2}, we need to assume that
$P_{Z|X}=Q_{Z|X}$ 
is also a BSC with parameter $\kappa'_1$. To define $P_{YZ|X}$, we need to couple variables $Y$ and $Z$ given $X$, consistent with both $P_{Y|X}$ and $P_{Z|X}$. 
We consider the following symmetric coupling:
$$
P_{YZ|X=0} : \quad 
\begin{array}{c|cc}
Y/Z & 0 & 1 \\ \hline
0 & 1 - \kappa_1' - \kappa_0 + \alpha_{11} & \kappa_1' - \alpha_{11} \\ 
1 & \kappa_0 - \alpha_{11} & \alpha_{11}
\end{array}
$$
$$
P_{YZ|X=1} : \quad
\begin{array}{c|cc}
Y/ Z & 0 & 1 \\ \hline
0 &  \alpha_{11} & \kappa_0 - \alpha_{11} \\
1 & \kappa_1' - \alpha_{11} & 1 - \kappa_1' - \kappa_0 +\alpha_{11}
\end{array}
$$
for some ${\alpha_{11} \in [\max(0,\kappa_0+\kappa'_1-1),\min(\kappa_0, \kappa_1')]}$. 
\begin{theorem}[Evaluation of Theorem \ref{RW_thm} from \cite{RahmanWagner2011}]
    For any $\rho_0 \in [-1,1]$ and $\rho_1 \in [0,1]$ and for extensions $P_{Z|XY}$ and $Q_{Z|XY}$ described above, Rahman and Wagner's bound in Theorem \ref{RW_thm} evaluates to:
\begin{align}
    & \min_{\alpha_{11} \in [0,\min(\kappa_0, \kappa_1')], \kappa_1' \in [0, \kappa_1]} \nonumber
    \\&  \left \{D(P_{YZ} \| Q_{YZ}) + H_P(Y|Z) + \min_{\lambda \in [0,1]} \Big \{   \lambda R - \lambda H_P(X|Z) + \max_{a \in [0,1]}[\lambda H_r(X|Z) - H_r(Y|Z)] \Big \} \right \}
\end{align}
where the distribution
$r_{XYZ}$ is defined as follows: $r(X=1) \in [0,1]$ and $r_{XYZ} = r_X P_{YZ|X}$. 

\label{thm_discrete_RW}
\end{theorem}

\begin{proof}
Using the alternative characterization of the Rahman-Wagner upper bound in Lemma \ref{lemmaaltern3} given in Appendix \ref{appendixlemmaaltern3}, we need to minimize over $(P_{Z|XY}, Q_{Z|XY})$ of   
\begin{align}
    D(P_{YZ} \| Q_{YZ}) + H_P(Y|Z) + \min_{\lambda \in [0,1]} \left \{   \lambda R - \lambda H_P(X|Z) + \max_{P_{U|X}}  [ \lambda H_P(X|U,Z) -H_P(Y|U,Z)] \right \}.
\end{align}

The inner maximization over $P_{U|X}$ evaluates the concave envelope of $r(x)\mapsto \lambda H_r(X|Z) - H_r(Y|Z)$ at $p(X=1) = \frac12$. Due to the symmetric structure of the channels $P_{YZ|X=x}$ for $x=0,1$, we can utilize the symmetrization argument of \cite{nair2013}: the concave envelope attains its maximum at $P(X=1) = \frac12$, and the value of the concave envelope
is equal to the maximum of the function  $r(x)\mapsto \lambda H_r(X|Z) - H_r(Y|Z)$. 
\end{proof}

\subsubsection{New Bound} To evaluate the bound in Theorem \ref{main_Thm}, we take the channels $P_{Z|X}$ and $Q_{Z|X}$ to be BSC with parameters $\kappa'_0$ and $\kappa'_1$ respectively. Consider the constraint on $P_{U|X}$: \begin{align}
    R\geq \max(I_P(U;X|Z), I_P(U;X|Y))\label{constonRs}.
\end{align}
Note that
$I(U;X|Z)=I(U;X)-I(U;Z)$ and $I(U;X|Y)=I(U;X)-I(U;Y)$. 
Since both  $P_{Y|X}$ and $P_{Z|X}$ are BSC with parameters $\kappa_0$ and $\kappa'_0$ respectively, depending on the channel parameters, one is always less noisy than the other. More specifically, if $$\left|\kappa'_0-\frac12\right|\geq \left|\kappa_0-\frac12\right|$$
we will have, $I(U;Z)\geq I(U;Y)$ for all $p_{U|X}$. Hence,
$$I(U;X|Z)\leq I(U;X|Y), \qquad\forall P_{U|X}.$$
We call this  \emph{case (a)}, and show the channel output by $Z_a$ and the channel parameter by $\kappa'_{0a}$ in this case.
On the other hand, if
$$\left|\kappa'_0-\frac12\right|\leq \left|\kappa_0-\frac12\right|$$
we will have, $I(U;Z)\leq I(U;Y)$ for all $p_{U|X}$. Hence,
$$I(U;X|Z)\geq I(U;X|Y), \qquad\forall P_{U|X}.$$
We call this  \emph{case (b)}, and show the channel output by $Z_b$ and the channel parameter by $\kappa'_{0b}$ in this case.

\begin{theorem}[Evaluation of Theorem \ref{main_Thm}]
    For any $\rho_0 \in [-1,1]$ and $\rho_1 \in [0,1]$ in the DSBS setting, the bound in Theorem \ref{main_Thm} yields the following upper bound:
    $$
    E_{P_{XY}, Q_{XY}}(R) \leq \min \lbrace G_{P_{XY}, Q_{XY}, P_{XZ_a}, Q_{XZ_a}}(R)+ D(P_{XZ_a} \| Q_{XZ_a}),  G_{P_{XY}, Q_{XY}, P_{XZ_b}, Q_{XZ_b}}(R)+ D(P_{XZ_b} \| Q_{XZ_b})\rbrace 
    $$
    where
    $$
    G_{P_{XY}, Q_{XY}, P_{XZ_a}, Q_{XZ_a}}(R) = \min_{\substack{\kappa_{0a}'\in \mathcal A\\ \kappa_1' \in [0,1]\\ \lambda \geq 0 }} \left \{ \lambda R - \lambda H_P(X|Y) +\max_{\substack{r_X,\hat Q_X}}  \left \{ [D(r_Y \| \hat Q_{Y}) - D(r_{Z_a} \| \hat Q_{Z_a} ) ]  +\lambda H_r(X|Y)\right \}\right \},
    $$
    $$
    G_{P_{XY}, Q_{XY}, P_{XZ_b}, Q_{XZ_b}}(R) = \min_{\substack{\kappa_{0b}'\in \mathcal B\\ \kappa_1' \in [0,1]\\ \lambda \geq 0 }} \left \{ \lambda R - \lambda H_P(X|Z_b) +\max_{\substack{r_X,\hat Q_X}}  \left \{ [D(r_Y \| \hat Q_{Y}) - D(r_{Z_b} \| \hat Q_{Z_b} ) ]  +\lambda H_r(X|Z_b)\right \}\right \},
    $$
where 
$$\mathcal{A}=\left\{\kappa'_{0a}:\left|\kappa'_{0a}-\frac12\right|\geq \left|\kappa_0-\frac12\right|\right\}, \quad\mathcal{B}=\left\{\kappa'_{0b}:\left|\kappa'_{0b}-\frac12\right|\leq \left|\kappa_0-\frac12\right|\right\}$$
and the distributions are defined as follows: 
$P_{Z_a|X} \sim BSC(\kappa_{0a}')$, $Q_{Z_a|X} \sim BSC(\kappa_1'),\  \kappa_1' \in [0,1]$. 
$r(X=1) \in [0,1]$, $r_{XY} = r_X P_{Y|X}$, $r_{XZ_a} = r_XP_{Z_a|X}$; $\hat Q_{X,Y} = \hat Q_X Q_{Y|X}$, $\hat Q_{X,Z_a} = \hat Q_X Q_{Z_a|X}$. Similarly, 
$P_{Z_b|X} \sim BSC(\kappa_{0b}')$, $Q_{Z_b|X} \sim BSC(\kappa_1'),\  \kappa_1' \in [0,1]$. 
$r(X=1) \in [0,1]$, $r_{XY} = r_X P_{Y|X}$, $r_{XZ_b} = r_XP_{Z_b|X}$; $\hat Q_{X,Y} = \hat Q_X Q_{Y|X}$, $\hat Q_{X,Z_b} = \hat Q_X Q_{Z_b|X}$.
\label{thm_discrete_main}
\end{theorem}

\begin{proof}
In case (a), the constraint in \eqref{constonRs} can be written as
$$R\geq I(U;X|Y).$$
The difference term $G_{P_{XY}, Q_{XY}, P_{XZ_a}, Q_{XZ_a}}(R)$ can be expressed as follows:
\begin{align}
    &\min_{\substack{P_{Z_a|X}\\ Q_{Z_a|X}}} \max_{\substack{P_{U|X}\in\mathcal{T}(P_{XY},P_{XZ_a})}}  \left \{\sum_u p(u) \max_{\hat Q_X} [D(P_{Y|u} \| \hat Q_{Y}) - D(P_{Z_a|u} \| \hat Q_{Z_a} ) ] \right \}
    \\&=\min_{\substack{P_{Z_a|X}\\ Q_{Z_a|X}}} \max_{\substack{P_{U|X}}} \min_{\lambda \geq 0} \left \{\sum_u p(u) \max_{\hat Q_X} [D(P_{Y|u} \| \hat Q_{Y}) - D(P_{Z_a|u} \| \hat Q_{Z_a} ) ] + \lambda [R - I_P(U;X|Y)] \right \}    \\
    &\leq \min_{\substack{P_{Z_a|X}\\ Q_{Z_a|X}}} \min_{\lambda \geq 0}\max_{\substack{P_{U|X}}}  \left \{\sum_u p(u) \max_{\hat Q_X} [D(P_{Y|u} \| \hat Q_{Y}) - D(P_{Z_a|u} \| \hat Q_{Z_a} ) ] + \lambda [R - I_P(U;X|Y)] \right \} \label{eq_discrete_minimax}
    \\&=\min_{\substack{P_{Z_a|X}\\ Q_{Z_a|X}}} \min_{\lambda \geq 0}\left[\lambda R- \lambda H_P(X|Y)+\max_{\substack{P_{U|X}}}  \left \{\sum_u p(u) \max_{\hat Q_X} [D(P_{Y|u} \| \hat Q_{Y}) - D(P_{Z_a|u} \| \hat Q_{Z_a} ) ] + \lambda \sum_u p(u) H(X|Y,U=u)\right \}\right]
    \\&= \min_{\substack{P_{Z_a|X}\\ Q_{Z_a|X}}} \min_{\lambda \geq 0}\left[\lambda R- \lambda H_P(X|Y)+\max_{r_X}  \left \{ \max_{\hat Q_X} [D(r_{Y} \| \hat Q_{Y}) - D(r_{Z_a} \| \hat Q_{Z_a} ) ] + \lambda H_r(X|Y)\right \}\right].\label{eq_concave_sym}
\end{align}
Note that the inequality \eqref{eq_discrete_minimax} follows from the Max-Min inequality; equation \eqref{eq_concave_sym} follows because the inner maximization evaluates the concave envelope of $r_X\mapsto\max_{\hat Q_X} [D(r_Y \| \hat Q_{Y}) - D(r_Z \| \hat Q_{Z} ) ]  +\lambda H_r(X|Y)$ at the uniform distribution $p(X=1)=\frac{1}{2}$. Using the symmetrization argument in \cite{nair2013},  the concave envelope attains its maximum at $P(X=1) = \frac12$ and its value is given by the maximum value of the map $r_X\mapsto\max_{\hat Q_X} [D(r_Y \| \hat Q_{Y}) - D(r_Z \| \hat Q_{Z} ) ]  +\lambda H_r(X|Y)$.

The second case is similar. 
\begin{figure}[t]
    \centering
    \includegraphics[width=1\textwidth]{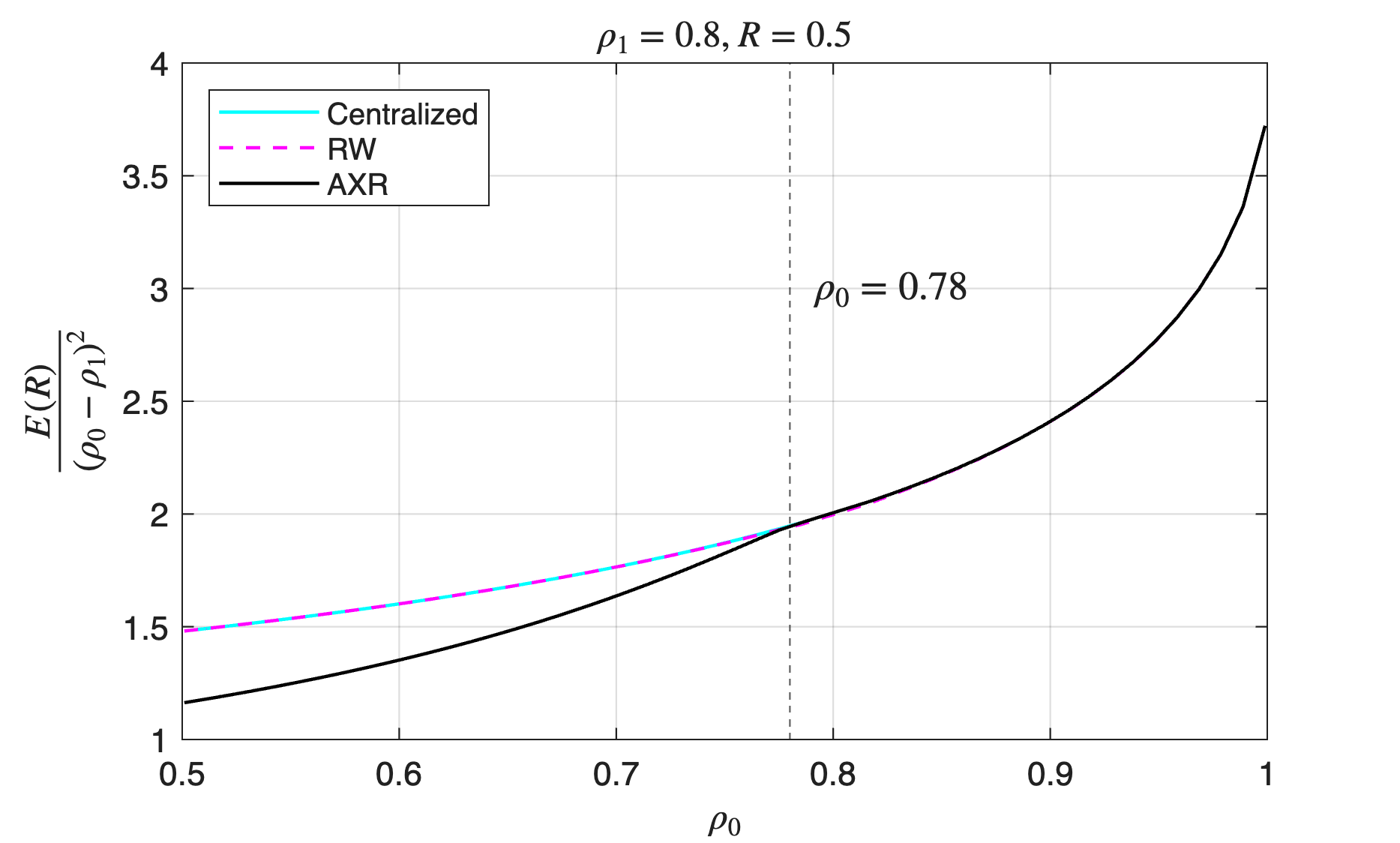}
    \caption{DSBS sources in the high rate regime: exponents (normalized by \((\rho_0 - \rho_1)^2\)) as a function of \(\rho_0\) for fixed \(\rho_1 = 0.8\) and rate \(R = 0.5\). The centralized bound is tight when $\rho_0$ exceeds $0.78...$ in this plot. The Rahman-Wagner upper bound matches the centralized bound for these parameters.}
    \label{fig_discrete_2}
\end{figure}

For the terms $E_{P_{XZ_a}, Q_{XZ_a}}(R)$ and $E_{P_{XZ_b}, Q_{XZ_b}}(R)$, we apply the centralized bound in \eqref{centralized_bound}. The proof is now complete.

\end{proof}

Plots comparing the two bounds in Theorem \ref{thm_discrete_RW} and Theorem \ref{thm_discrete_main} are given in Figures \ref{fig_discrete_2}-\ref{fig_against_rate}. The error exponents are normalized by $(\rho_0 - \rho_1)^2$ for better visualization. In all of the plots, the Rahman and Wagner's upper bound in \cite{RahmanWagner2011} as evaluated in Theorem \ref{thm_discrete_RW}, and the new bound as evaluated in Theorem \ref{thm_discrete_main} are depicted.

{

Fig. \ref{fig_discrete_2} is a high-rate case.
In Fig. \ref{fig_discrete_2}, we fix $R=0.5$, and $\rho_1 = 0.8$ (or $\kappa_1=0.1$).
The centralized bound is tight when 
$$R\geq D(P_{XY}\|Q_{XY})+H_P(X|Y)=(1 - \kappa_0) \log_2 \left( \frac{1}{1 - \kappa_1} \right) + \kappa_0 \log_2 \left( \frac{1}{\kappa_1} \right).$$
Moreover, by Theorem \ref{thm-general-position}, it is not tight when
$R<H_2(\kappa_0)$
where $H_2$ is the binary entropy function. Equivalently, in our example, the centralized bound is tight when
$$\kappa_0\leq \frac{R+\log_2(1-\kappa_1)}{\log_2(1-\kappa_1)-\log_2(\kappa_1)}=0.10978...$$
This corresponds to $\rho_0\geq 0.78043...\approx 0.78$.
The new bound  improves over the centralized bound when
$$\kappa_0>H_2^{-1}(0.5)\approx 0.11002...$$
This corresponds to $\rho_0\leq 0.77994...\approx 0.78$.
Fig. \ref{fig_discrete_2}, shows that both Rahman and Wagner's upper bound and the new bound reduce to the centralized bound after $\rho_0 > 0.78...$. Note that the new bound yields an improvement for $\rho_0 < 0.78...$ with respect to the centralized bound, whereas Rahman and Wagner's upper bound does not. 

Fig. \ref{fig_discrete_1} is a low-rate case. Fig. \ref{fig_discrete_1} plots the bounds in a low rate regime where we set $R=0.1$ and $\rho_1=0.7$. We see that the new bound has a kink because for $\rho_0\leq 0.92...$ case (a) in Theorem \ref{thm_discrete_main} yields the best bound, while for $\rho_0 \geq 0.92...$ case (b) yields the best bound.

Fig. \ref{fig_against_rate} plots the two cases in Theorem \ref{thm_discrete_main} along with the Rahman-Wagner bound and the centralized bound under different rates while fixing $\rho_0 = 0.95$ and $\rho_1 = 0.8$. Both the new bound and the Rahman-Wagner bound approach zero when the rate goes to zero. The case (b) in Theorem \ref{thm_discrete_main} is always below the Rahman-Wagner bound for these parameters. The case (a) in Theorem \ref{thm_discrete_main} beats case (b) after $R > 0.08...$.

\begin{remark}
    For Fig. \ref{fig_discrete_2}, \ref{fig_discrete_1}, and \ref{fig_against_rate}, we numerically observe that case (a) in the new bound reduces to the one given in \eqref{rhs19eq}  with the choice of $Z_a=X$ under both $P$ and $Q$.
\end{remark}
}

\begin{figure}[t]
    \centering
    \includegraphics[width=1\textwidth]{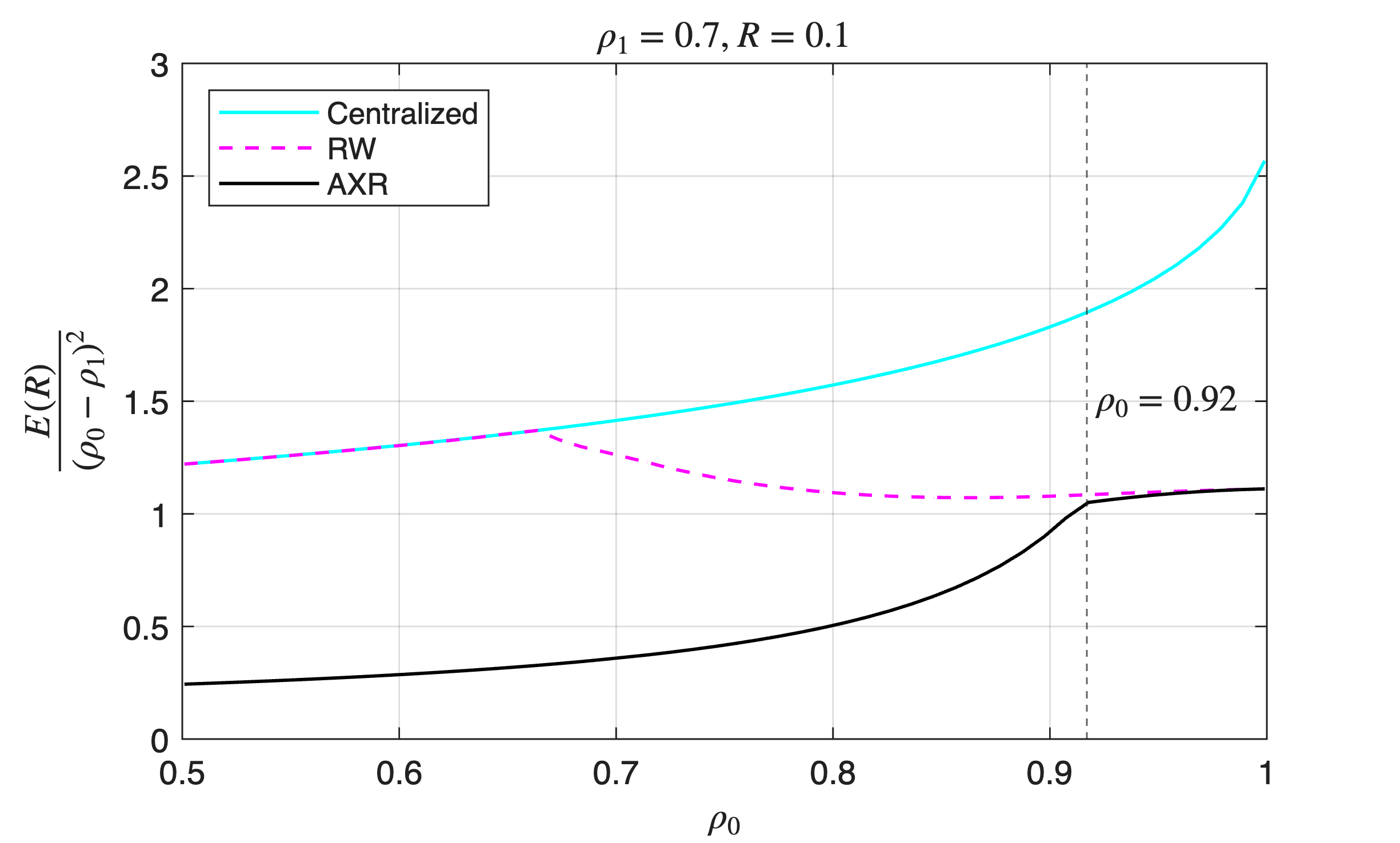}
    \caption{DSBS sources in the low rate regime: exponents (normalized by \((\rho_0 - \rho_1)^2\)) as a function of \(\rho_0\) for fixed \(\rho_1 = 0.7\) and rate \(R = 0.1\). In this case, $R = 0.1 <H_P(X|Y)$ for $\rho_0 \in [0.5,1]$. By Theorem \ref{thm-general-position}, the centralized bound is not tight. The new bound beats both the centralized bound and the Rahman-Wagner bound.}

\label{fig_discrete_1}
\end{figure}

\begin{figure}[h]
    \centering
    \includegraphics[width=1\textwidth]{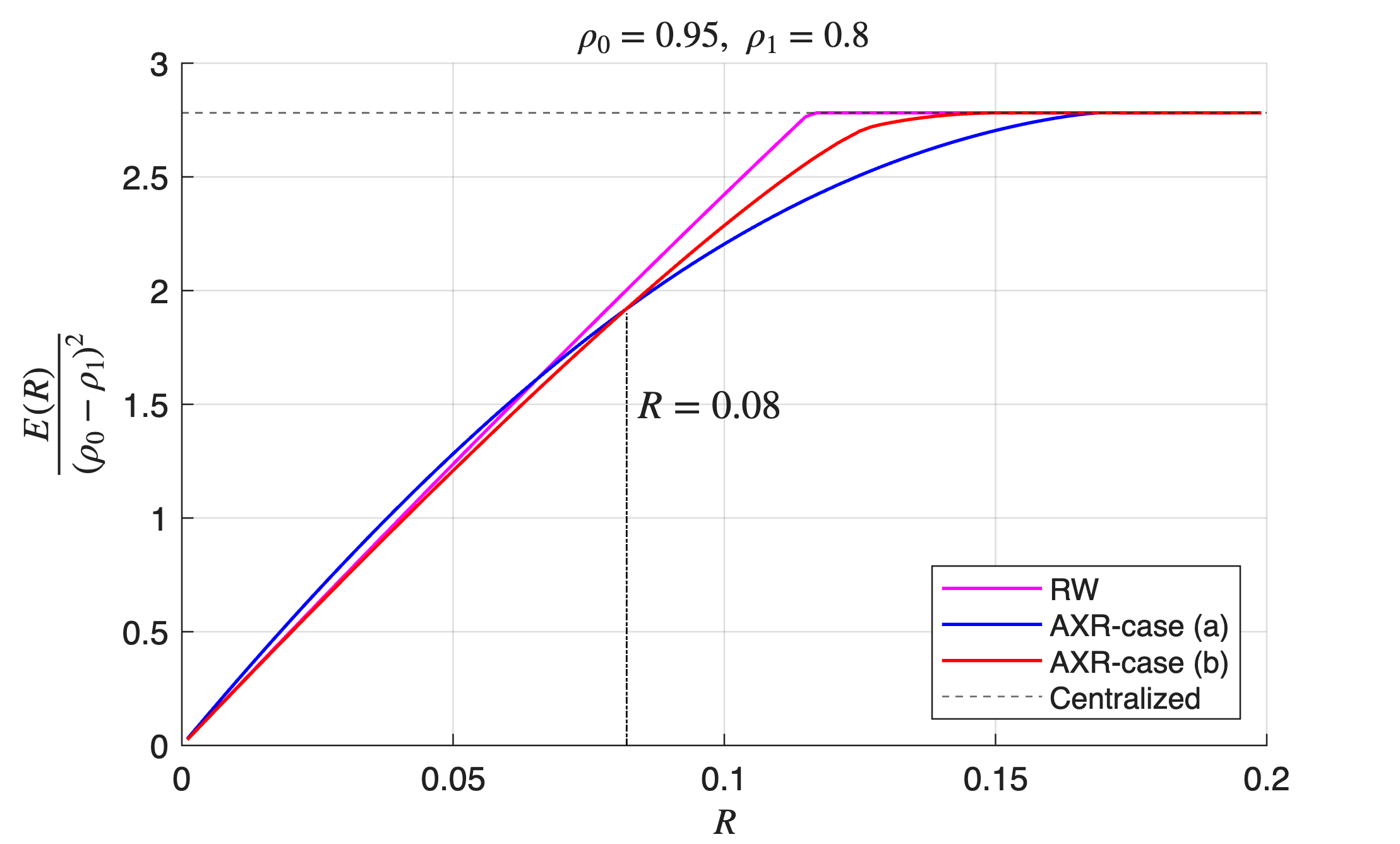}
    \caption{A DSBS source in different rate regimes: exponents (normalized by \((\rho_0 - \rho_1)^2\)) as a function of rate \(R\) for fixed \(\rho_0 = 0.95\) and \(\rho_1 = 0.8\). The new bound, which is the minimum of case (a) and (b), beats the centralized bound and the Rahman-Wagner bound.}

\label{fig_against_rate}
\end{figure}

\subsection{Strict improvement for Gaussian Sources}

Next, we show that the bound in Corollary \ref{corr-weaker} can provide better bounds than the one in \cite{RahmanWagner2011}. 
Authors in \cite{RahmanWagner2011} evaluate their bound for the following Gaussian setting: consider a Gaussian scenario where \(X\) and \(Y\) are jointly Gaussian sources with zero mean and unit variance. Assume that their correlation coefficients under hypotheses \(H_0\) and \(H_1\) are \(\rho_0\) and \(\rho_1\), respectively. In other words, under the hypothesis $H_i$, we have \begin{align*}
 & K_{X,Y} \sim \mathcal{N}\bigg(0,\begin{bmatrix} 
1&\rho_i\\\rho_i&1 
\end{bmatrix}\bigg) 
 \end{align*} 
where \(\rho_0 \in [-1,1]\), and without loss of generality, \(\rho_1 \in [0,1]\).
Even though we stated our bounds for discrete settings, a similar argument goes through for jointly Gaussian random variables. Moreover, in this section, we measure the communication rate in nats per symbol, the logarithms and the exponent are in base $e$.

\subsubsection{Existing Bounds}
The authors in \cite{RahmanWagner2011} and  \cite{hadar2019error} provide the following bounds:
\begin{theorem}[Theorem 9 in \cite{RahmanWagner2011}]
    For the Gaussian setting, we have
   \begin{align*}
    E_{RW}(R) &=
    \frac{1}{2} \log \left ( \frac{1}{1 - {\rho}^2 + {\rho}^2 e^{-2R}}\right) + C
\end{align*}
    where
% \begin{align}
% \rho &=
%   \frac{\rho_0 - \rho_1}{1 - \rho_1}.
% \end{align}
\begin{align*}
\mathcal{D}_1 &\triangleq \{(\rho_0,\rho_1) : 0 \le \rho_1 < \rho_0 < 1\}, \\
\mathcal{D}_2 &\triangleq \{(\rho_0,\rho_1) : 0 \le \rho_1 \hspace{0.05in}\textrm{and}\hspace{0.05in} 2 \rho_1 - 1 \le \rho_0 < \rho_1\}, \\
\mathcal{D}_3 &\triangleq \Big \{(\rho_0,\rho_1) : -1 < \rho_0 \le 2 \rho_1 - 1  \hspace{0.05in}\textrm{and}
\\&\hspace{0.05in} \frac{ 2 \rho_1}{1-\rho_1} \le \frac{1}{2} \log \left(\frac{1-\rho_1^2}{1-\rho_0^2}\right) - \frac{ \rho_1(\rho_0-\rho_1)}{1-\rho_1^2} \Big\}.
\\&\rho \triangleq \left\{
\begin{array}{l l}
  \frac{\rho_0 - \rho_1}{1 - \rho_1} & \quad \mbox{if $(\rho_0,\rho_1)$ is in $\mathcal{D}_1 \cup \mathcal{D}_2$},\\
  \frac{\rho_0 + \rho_1}{1 - \rho_1} & \quad \mbox{if $(\rho_0,\rho_1)$ is in $\mathcal{D}_3$}. \\ \end{array}\label{defrho} \right.
\\&C \triangleq \left\{
\begin{array}{l l}
  0 & \quad \mbox{if $(\rho_0,\rho_1)$ is in $\mathcal{D}_1 \cup \mathcal{D}_2$},\\
  \frac{2 \rho_1}{1 - \rho_1} & \quad \mbox{if $(\rho_0,\rho_1)$ is in $\mathcal{D}_3$}. \\ \end{array} \right.
\end{align*}
\label{theoremR}
\end{theorem}

\begin{theorem}[Theorem 1 in \cite{hadar2019error}]
    For any \(\rho_0, \rho_1 \geq 0\),
    \begin{align*}
        E_{HLPS}(R) = \frac{R}{\Big(\frac{1 - \min\{\rho_0, \rho_1\}}{\rho_1 - \rho_0} \Big)^2 - 1}.
    \end{align*}
\end{theorem}

In terms of lower bounds, the expression of SHA exponent in the Gaussian setting was given in \cite{Yuval2024}:
\begin{theorem}[Equations (35)-(40) in \cite{Yuval2024}] For the Gaussian hypothesis testing the following exponent is achievable:
\begin{align*}
    & E_{SHA}(R) = \sup_{\eta \geq 0} \max \Big [ E_{\eta}(R), \max_{0 \leq \beta \leq \beta_{\max}}
    \min \{ E_{\eta}(\beta R), \overline{E}_\eta(\beta R) - (\beta - 1)R + E_\eta (\beta R) {\mathbf{1}}_{\{ \rho_0 < \rho_1 \}} \} \Big]
\end{align*}
where 
\begin{align*}
    &E_{\eta}(R) = D \Big ( \frac{(\rho_1 - \eta)^2 \sigma_R^2 + \rho_1^2(1 - \sigma_R^2)}{1 + \eta(\eta - 2\rho_1)\sigma_R^2}, \frac{1 + \eta(\eta - 2\rho_0)\sigma_R^2}{1 + \eta(\eta - 2\rho_1)\sigma_R^2}  \Big ),
    \\&\overline{E}_\eta (R) = D \Big ( \frac{\eta^2 \sigma_R^2}{1 + \eta^2 \sigma_R^2}, \frac{1 + \eta(\eta - 2\rho_0)\sigma_R^2}{1 + \eta^2 \sigma_R^2} \Big ),
    \\& \sigma_R^2 = 1 - \exp(-2R),
    \\& D(\gamma, P) = \frac{\gamma + P - A(\gamma, P)}{2(1 - \gamma)} - \frac{1}{2} \log \frac{A(\gamma, P) - 1 + \gamma}{2 \gamma},
    \\& A(\gamma, P) = \sqrt{(1 - \gamma)^2 + 4P\gamma}.
\end{align*}
The maximization limit \(\beta_{\max}\) is the unique solution of \(\overline{E}_\eta (\beta R) = (\beta - 1) R\).
\end{theorem}

\subsubsection{New Bound}
Below, we state our new bound by optimizing the lowest upper bound using Corollary \ref{corr-weaker} (which is a weaker version of our main bound in Theorem \ref{theorem_1}):
\begin{theorem}
    For the standard Gaussian setting and assuming $\rho_0 \in [-1,1], \rho_1 \in [0,1]$, we have
\begin{align*}
    E_{P_{Z|XY},Q_{Z|XY}}(R) \leq 
\min_{(P_{Z|XY}, Q_{Z|XY}) \in \mathcal{R}} &\Big \{
\max_{P_{U|X} \in \mathcal{T}_R} I_P(Y;U | Z) + D(P_{YZ} \| Q_{YZ})\Big \} , \end{align*}
where
    \begin{align*}
        \mathcal{R} &= \big \{ (P_{Z|XY}, Q_{Z|XY}): Q_{YZ|X} = Q_{Z|X}Q_{Y|Z}, P_{XZ} = Q_{XZ} \big \},\\
        \mathcal{T}_R &= \left \{ P_{U|X} : R \geq \max(I_P(U;X|Z), I_P(U;X|Y) \right \}.
    \end{align*}
    For evaluation of the bound, one should use the following expression:
\begin{align*}
    E_{P_{XY},Q_{XY}}(R) \leq 
\min_{b,g} &\Big \{ I_P(Y;U | Z) 
+ D(P_{YZ} \| Q_{YZ})\Big \} , \end{align*}
where\begin{align*}
    \hat{\sigma}^{2} &= \max \left ( 
    \frac{1 - g^2}{\exp(2R) - g^2} ,
    \frac{1 - (a + bg)^2}{\exp(2R) - (a + bg)^2}
    \right),
\\ I_P(Y;U | Z) &= \frac{1}{2} \log
\left 
    \{
    \frac
    {1 - [(\rho_0 - bg)g + b]^2}
    {1 - \rho_0^2 (1 - \hat\sigma^2) - \frac{[\rho_0 g \hat\sigma^2 + b (1 - g^2)]^2}
    {1 - (1 - \hat\sigma^2)g^2}}
    \right 
    \}, \\
 D(P_{YZ} \| Q_{YZ}) &= \frac{g^2 - {\rho_1}{g}[(\rho_0 - bg)g + b] }{{g^2} - {\rho_1^2}} - 1 + \frac{1}{2} \log \frac{{g^2} - {\rho_1^2}}{{g^2} - {g^2}[(\rho_0 - bg)g + b]^2},\end{align*}
where the minimum is over $b,g$ satisfying
\begin{align*}
& g \in [-1,1],\quad  b^2(1 - g^2) \leq 1 - \rho_0^2, \quad ((\rho_0 - bg)g + b)^2 \leq 1, \quad  \rho_1^2 \leq g^2. 
\end{align*}
\label{theoremG}
\end{theorem}
Plots comparing the bound with all the other known bounds are given in Fig.~\ref{fig_compare_1} and Fig.~\ref{fig_compare_2}. One can observe that our bound is always better than the bound in \cite{RahmanWagner2011} for these parameters. Computing the bound in \cite{Yuval2024} is difficult, so we did not replicate the simulation done by \cite{Yuval2024} and just extracted the points from the curve given in \cite{Yuval2024}. One can observe that the auxiliary receiver bound is better than the bound in \cite{Yuval2024} when $\rho_0$ is not close to one. For $\rho_0$ close to one, the bound in \cite{Yuval2024} is better, which suggests multiple auxiliary receivers can yield strictly better bounds for large $\rho_0$. We have discussed combining the idea in \cite{Yuval2024} with our technique in Section \ref{sec31}.

\begin{figure}[h]
    \centering
    \includegraphics[width=1\textwidth]{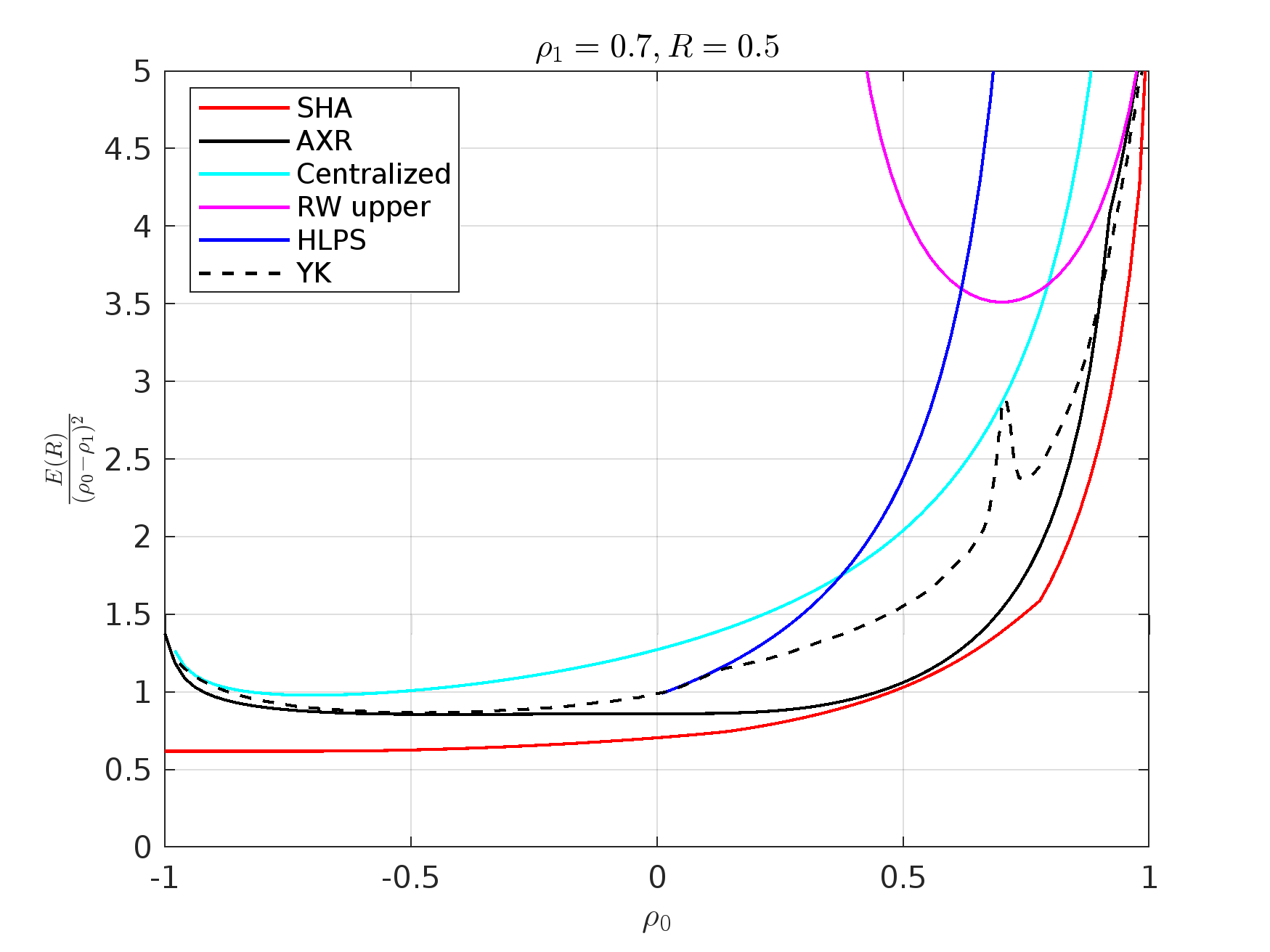}
    \caption{Exponents (normalized by \((\rho_0 - \rho_1)^2\)) as a function of \(\rho_0\) for fixed \(\rho_1 = 0.7\) and rate \(R = 0.5\). Red is SHA's achievable exponent. Cyan is the centralized upper bound. Blue is HLPS upper bound in \cite{hadar2019error}. Magenta is RW upper bound in \cite{RahmanWagner2011}. Black is our bound in Theorem \ref{theoremG}. Dashed black is the upper bound extracted from the plot in \cite{Yuval2024}.}
    \label{fig_compare_1}
\end{figure}

\begin{figure}[h]
    \centering
    \includegraphics[width=1\textwidth]{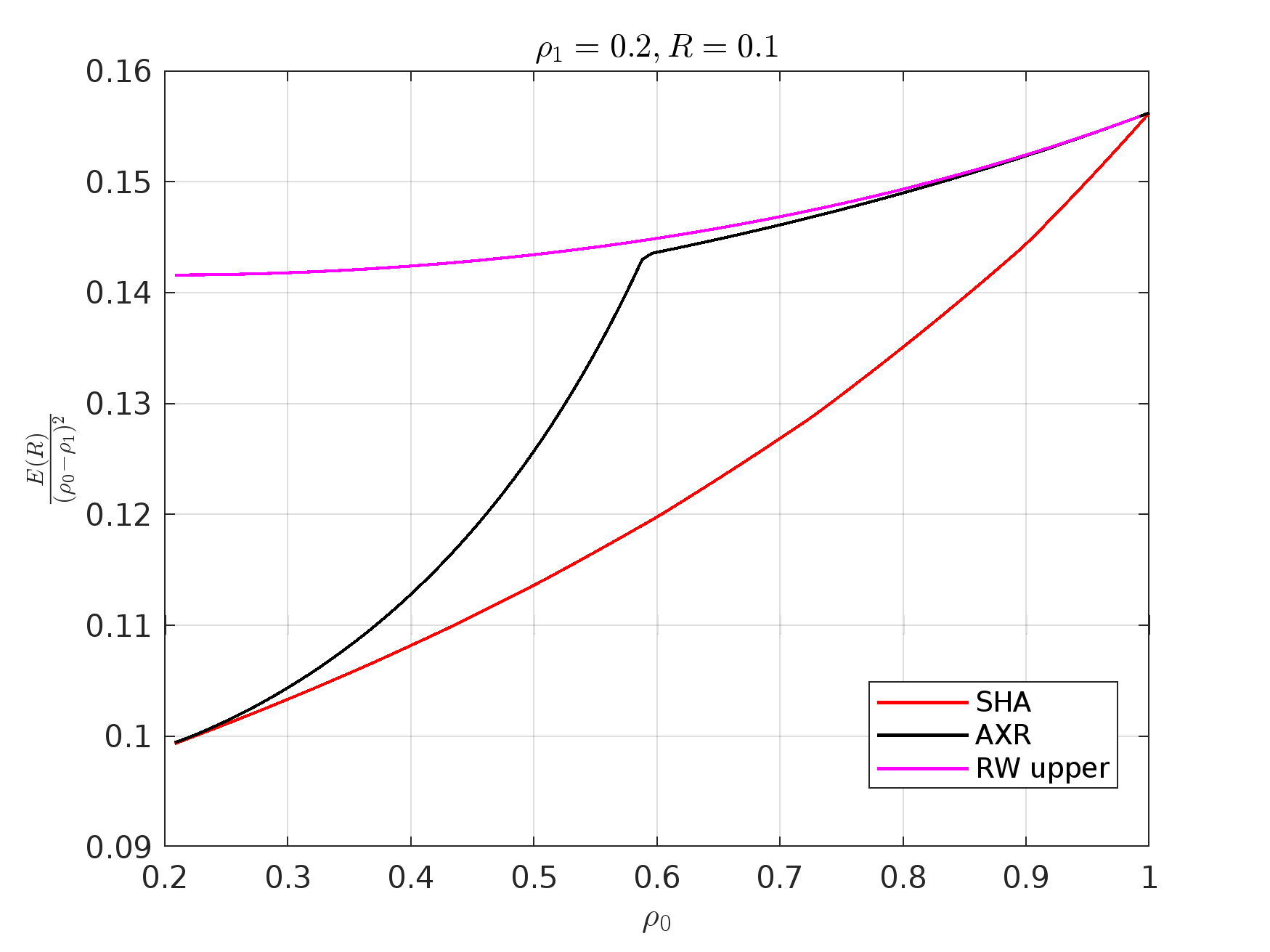}
    \caption{Exponents (normalized by \((\rho_0 - \rho_1)^2\)) with fixed \(\rho_1 = 0.2\) and rate \(R = 0.1\). In this case, the other bounds discussed are large and thus not shown here. \(\rho_0 \in [\rho_1, 1]\) for a better visual illustration with RW upper bound in \cite{RahmanWagner2011}. AXR (auxiliary receiver) upper bound goes to RW upper bound as \(\rho_0 \rightarrow 1\) and to the achievable SHA exponent as \(\rho_0 \rightarrow \rho_1\).}
    \label{fig_compare_2}
\end{figure}

\begin{proof}[Proof of Theorem \ref{theoremG}]

Under $P$, assume that
$$Z=g X+\sqrt{1-g^2} N_1,$$
$$Y=aX+bZ+N_3\sqrt{1 - a^2 - b^2-2abg},$$
and under $Q$, assume that
$$Z=g X+\sqrt{1-g^2} N_1,$$
$$Y=fZ+\sqrt{1-f^2} N_2,$$
where $I(N_1;X)=I(N_2;Z)=I(N_3;X,Z)=0$ and $X,Y,Z,N_1,N_2,N_3$ are standard Gaussian random variables.

Then we remove redundant parameters using \(\rho_0\) and \(\rho_1\) and obtain constraints on the free parameters.
Under \(P\), the covariance of \(X\) and \(Y\) is \(\rho_0\):
\begin{align*}
    &Cov(X,Y) = a + bg = \rho_0,\\
    &Cov(Y,Z) = ag + b \in [-1,1], \\
    &Cov(Y, N_2) = \sqrt{1 - (a^2 + b^2 + 2abg)} \in [-1,1], \\
    &Cov(X,Z) = g \in [-1,1].
\end{align*}
Under \(Q\), the covariance of \(X\) and \(Y\) is \(\rho_1\), thus:
\begin{align*}
    &Cov(X,Y) = Cov(X, fgX) = fg = \rho_1,\\
    &Cov(Y,Z) = f \in [-1,1].
\end{align*}

The above characterization gives
\begin{align*}
    &I_P(Y;U |Z) =  \frac12 \log \left 
    \{
    \frac
    {1 - (ag + b)^2}
    {1 - (a + bg)^2 (1 - \sigma^2) - \frac{[(a+bg)g\sigma^2 + b (1 - g^2)]^2}
    {1 - (1 - \sigma^2)g^2}}
    \right 
    \},
\end{align*}

and 
\[
D(P_{YZ} \| Q_{YZ}) = \frac{1 - f(ag+b)}{1 - f^2} - 1 + \frac12 \log \frac{1 - f^2}{1 - (ag+b)^2}.
\]
We show the optimality of a jointly Gaussian distribution on $(U,X)$ when evaluating $G_{P_{XY},Q_{XY},P_{XZ},Q_{XZ}}(R)$ (and hence the above weaker upper bound) in Appendix \ref{appGaussian}. 

If the variance of \(P_{X|U}\) is \(\sigma^2 \leq 1\). The derivative of the denominator of \(I_P(Y;U|Z)\) is $\frac{a^2 (g^2 - 1)^2}{(g^2 \sigma^2 - g^2 + 1)^2} \geq 0$. The optimal \(\sigma^2\) that maximizes \(I_P(Y;U|Z)\) within feasible region \(\mathcal{T}_R\) is then given by:
\begin{align*}
    \hat{\sigma}^{2} = \max \left ( 
    \frac{1 - g^2}{\exp(2R) - g^2} ,
    \frac{1 - (a + bg)^2}{\exp(2R) - (a + bg)^2}
    \right).
\end{align*}

\end{proof}

\subsection{The bound in \cite{Yuval2024} and a sequence of auxiliary receivers}\label{sec31}
\begin{definition}[\cite{Yuval2024}]\label{definition_1b}
    Let \(\mathcal{\tilde{R}}(P_{XY}, Q_{XY})\) be the set of all \((P_{Z|XY}, Q_{Z|XY})\) such that
\begin{align}
      Q_{YZ|X} &= Q_{Z|X} Q_{Y|Z}, \label{eqnDef2a} \\
     P_{Y|Z} &= Q_{Y|Z}. \label{eqnDef2b} 
\end{align} 
\end{definition}
The following bound was given in \cite{Yuval2024}, which is in the same spirit as the auxiliary receiver bound:

\begin{lemma}[Lemma 1 in \cite{Yuval2024}]\label{lemmaYuval}
    For any \((P_{Z|XY}, Q_{Z|XY})\) in \(\tilde{\mathcal{R}}(P_{XY}, Q_{XY})\):
\begin{align*}
 D(P_{M Y^n} \| Q_{M Y^n})
 &\leq \max_{P_{U|X}\in\mathcal{S}(P_{XYZ})}[I_P(Y;U|Z) ]+ D(P_{M Z^n} \| Q_{M Z^n}). %\\
 %&\leq \sum_{l=1}^L \mathop{\max}\limits_{P_{U|XYZ,l \in \mathcal{S}(P_{XYZ,l})}} I_{P_l}(U;Y | Z) + D(P_{MZ^n,L} \| Q_{MZ^n,L}) 
\end{align*} 

\end{lemma}
Observe that the bound stated in Theorem \ref{theorem_1} allows for arbitrary choices of $P_{Z|XY}$ and $Q_{Z|XY}$, while the one in Lemma \ref{lemmaYuval} is restricted to a class of auxiliary receivers \((P_{Z|XY}, Q_{Z|XY})\) in \(\tilde{\mathcal{R}}(P_{XY}, Q_{XY})\). 
Note that
\((P_{Z|XY}, Q_{Z|XY})\) in \(\tilde{\mathcal{R}}(P_{XY}, Q_{XY})\) requires  \(P_{Y|Z} = Q_{Y|Z}\). Therefore, 
\begin{align*}
    D(P_{YZU} \| P_{UZ} Q_{Y|Z})&=D(P_{YZU} \| P_{UZ} P_{Y|Z})=I_P(Y;U|Z).
\end{align*}
Theorem \ref{theorem_wagner} then implies that the bound in Theorem \ref{theorem_1} is less than or equal to the one in Lemma \ref{lemmaYuval}.

The author in \cite{Yuval2024} suggests the use of a sequence of auxiliary receivers $Z_1, Z_2, \cdots, Z_k$ to derive an improved bound. This approach can also be integrated with our method, which might result in further improvement. In other words, for any sequence of auxiliary receivers $Z_1, Z_2, \cdots, Z_k$ we can write
\begin{align}
 & E_{P_{XY},Q_{XY}}(R) \leq G_{P_{XY},Q_{XY},P_{XZ_1},Q_{XZ_1}}(R)+\sum_{j=1}^{k-1}G_{P_{XZ_j},Q_{XZ_j},P_{XZ_{j+1}},Q_{XZ_{j+1}}}(R)+E_{P_{XZ_k},Q_{XZ_k}}(R).\label{eqnAnn1}
\end{align}

To motivate \eqref{eqnAnn1}, consider the bound
\begin{align*}
 & E_{P_{XY},Q_{XY}}(R) \leq G_{P_{XY},Q_{XY},P_{XZ},Q_{XZ}}(R)+E_{P_{XZ},Q_{XZ}}(R).
\end{align*}
Another trivial bound is the centralized bound:
\begin{align*}
    E_{P_{XY},Q_{XY}}(R)\leq D(P_{XY}\|Q_{XY}).
\end{align*}
The following example shows that the centralized bound might be better in some cases:
\begin{example}\label{exampleinf}
    Assume that $P_{Y|X}, Q_{Y|X}, P_{Z|X}$ and $Q_{Z|X}$ are additive Gaussian noise channels with additive noise variances $\sigma_{P,Y|X}^2$, $\sigma_{Q,Y|X}^2$,$\sigma_{P,Z|X}^2$, $\sigma_{Q,Z|X}^2$ respectively. Assume that $X\sim \mathcal{N}(0,\sigma_X^2)$ under $P$. We claim that $f_{P_{Y|X}, P_{Z|X}, Q_{Y|X}, Q_{Z|X}}(P_X)=\infty$ if $\sigma_{Q,Y|X}<\sigma_{Q,Z|X}$.
    Assume that $X\sim \mathcal{N}(k,\sigma_X^2)$ under $\hat Q_X$. 
    Let
    \begin{align*}
\sigma_{P,Y}&=\sqrt{\sigma_X^2+\sigma_{P,Y|X}^2}, \qquad
\sigma_{P,Z}=\sqrt{\sigma_X^2+\sigma_{P,Z|X}^2},
\\
\sigma_{Q,Y}&=\sqrt{\sigma_X^2+\sigma_{Q,Y|X}^2},\qquad
\sigma_{Q,Z}=\sqrt{\sigma_X^2+\sigma_{Q,Z|X}^2}.
    \end{align*}
    
    Then, we get
    \begin{align*}
  &f_{P_{Y|X}, P_{Z|X}, Q_{Y|X}, Q_{Z|X}}(P_X) \\ &\geq  D(P_Y \| \hat{Q}_Y) - D(P_Z \| \hat{Q}_Z) 
  \\&=\log\frac{\sigma_{Q,Y}}{\sigma_{P,Y}}+\frac{\sigma_{P,Y}^2-\sigma^2_{Q,Y}+k^2}{2\sigma^2_{Q,Y}} -\log\frac{\sigma_{Q,Z}}{\sigma_{P,Z}}-\frac{\sigma_{P,Z}^2-\sigma^2_{Q,Z}+k^2}{2\sigma^2_{Q,Z}}.
\end{align*} 
Letting $k\rightarrow\infty$, we get that $f_{P_{Y|X}, P_{Z|X}, Q_{Y|X}, Q_{Z|X}}(P_X)=\infty$.
\end{example}
\begin{proposition}\label{thmJ}
Take some arbitrary $P_{J|XYZ}$ and $Q_{J|XYZ}$.
Assume that $P_{XYZJ}$ and $Q_{XYZJ}$, satisfy  $D(P_{XYJ}\|Q_{XYJ})<\infty$ and $D(P_{XZJ}\|Q_{XZJ})<\infty$. Let $Y'=(Y,J)$ and $Z'=(Z,J)$. Then, we have
\begin{align}
 & E_{P_{XY},Q_{XY}}(R)
 \leq 
G_{P_{XY'},Q_{XY'},P_{XZ'},Q_{XZ'}}(R) +E_{P_{XZ'},Q_{XZ'}}(R).\label{eqnJBound1}
\end{align}
\end{proposition}
\begin{proof}
Consider \eqref{eqnAnn1} with the  choice $Z_1=Y'$ and $Z_2=Z'$:
\begin{align*}
 &E_{P_{XY},Q_{XY}}(R) \\ &\leq G_{P_{XY},Q_{XY},P_{XY'},Q_{XY'}}(R) +G_{P_{XY'},Q_{XY'},P_{XZ'},Q_{XZ'}}(R)+E_{P_{XZ'},Q_{XZ'}}(R)
 \\&\leq G_{P_{XY'},Q_{XY'},P_{XZ'},Q_{XZ'}}(R)+E_{P_{XZ'},Q_{XZ'}}(R)
\end{align*}
as $G_{P_{XY},Q_{XY},P_{XY'},Q_{XY'}}(R)\leq 0$. 
\end{proof}
\begin{remark}
    If we set $J$ to be constant, both bounds recover the original bound. If we set $J=X$, we recover the centralized bound because 
\begin{align*}
  &f_{P_{Y'|X}, P_{Z'|X}, Q_{Y'|X}, Q_{Z'|X}}(P_X) \\&= \max_{\hat{Q}_X } D(P_{XY} \| \hat{Q}_{XY}) - D(P_{XZ} \| \hat{Q}_{ZX}) 
  \\
  &= \max_{\hat{Q}_X } D(P_{Y|X} \| \hat{Q}_{Y|X}|P_X) - D(P_{Z|X} \| \hat{Q}_{Z|X}|P_X)
  \\&= D(P_{Y|X} \| Q_{Y|X}|P_X) - D(P_{Z|X} \| Q_{Z|X}|P_X)
\end{align*} 
    where we used the fact that $\hat{Q}_{Z|X}=Q_{Z|X}$ and  $\hat{Q}_{Y|X}=Q_{Y|X}$. Observe that $f_{P_{Y'|X}, P_{Z'|X}, Q_{Y'|X}, Q_{Z'|X}}(P_X)$ is linear in $P_X$, and it is optimal to choose constant $U$ when evaluating $G_{P_{XY'},Q_{XY'},P_{XZ'},Q_{XZ'}}(R)$. We obtain that
\begin{align*}
    &G_{P_{XY'},Q_{XY'},P_{XZ'},Q_{XZ'}}(R)\\
    &=D(P_{Y|X} \| Q_{Y|X}|P_X) - D(P_{Z|X} \| Q_{Z|X}|P_X)\\
  &= D(P_{XY} \| Q_{XY}) - D(P_{XZ} \| Q_{XZ}).
\end{align*}
  Thus, from \eqref{eqnJBound1} and using $E_{P_{XZ'},Q_{XZ'}}(R)=D(P_{XZ} \| Q_{XZ})$ when $J=X$, we obtain
$$G_{P_{XY'},Q_{XY'},P_{XZ'},Q_{XZ'}}(R)+E_{P_{XZ'},Q_{XZ'}}(R)=D(P_{XY} \| Q_{XY})$$
  as desired.
\end{remark}

\appendices
\section{Alternative characterization of the upper bound of Rahman and Wagner}
\label{appendixlemmaaltern3}
The following alternative characterization of the Rahman-Wagner bound for discrete alphabet random variables is useful in our numerical evaluations: 
\begin{lemma}
    [Alternative characterization of the bound in Theorem \ref{RW_thm}] For discrete sources, the bound in Theorem \ref{RW_thm} can be equivalently written as
    $$
\min_{(P_{Z|XY}, Q_{Z|XY}) \in \mathcal R (P_{XY}, Q_{XY})} \left \{D(P_{YZ} \| Q_{YZ})+ \min_{\lambda \in[0,1]}\max_{P_{U|X}} [I_P(Y;U|Z) + \lambda(R - I_P(X;U|Z)) ] \right \}.
$$
where the maximum is over all $P_{U|X}$.\label{lemmaaltern3} 
\end{lemma}

\begin{proof}[Proof of Lemma \ref{lemmaaltern3}]
    The Rahman-Wagner bound can be expressed as
$$
\min_{(P_{Z|XY}, Q_{Z|XY}) \in \mathcal R (P_{XY}, Q_{XY})} \left \{D(P_{YZ} \| Q_{YZ})+ \max_{P_{U|X}}\min_{\lambda \geq 0} [I_P(Y;U|Z) + \lambda(R - I_P(X;U|Z)) ] \right \}.
$$
To prove the lemma, it suffices to show the following two equalities:
\begin{align}
\max_{P_{U|X}}\min_{\lambda \geq 0} [I_P(Y;U|Z) + \lambda(R - I_P(X;U|Z)) ]&=\min_{\lambda \geq 0}\max_{P_{U|X}} [I_P(Y;U|Z) + \lambda(R - I_P(X;U|Z)) ]  \label{step1show}\\&=\min_{\lambda\in[0,1]}\max_{P_{U|X}} [I_P(Y;U|Z) + \lambda(R - I_P(X;U|Z)) ].\label{step2show}
\end{align}
We first show the equality in \eqref{step2show}. For any $\lambda \geq 1$,
\begin{align*}
    & I(Y;U|Z) - I(X;U|Z) + \lambda R - (\lambda - 1) I(X;U|Z) \\
    &= -I(X;U|Y,Z) + \lambda R - (\lambda - 1) I(X;U|Z) \\
    &\leq \lambda R.
\end{align*}
Moreover, $\lambda R$ is achieved when $U$ is a constant random variable. Therefore, for any $\lambda\geq 1$, we have
$$\max_{P_{U|X}} [I_P(Y;U|Z) + \lambda(R - I_P(X;U|Z))]=\lambda R$$
and is minimized when $\lambda=1$. 

It remains to show
the equality in \eqref{step1show}. Let
$$\mathcal{A}=\{(a_1,a_2): a_1=I(X;U|Z), a_2=I(Y;U|Z), \quad\text{for some }p(u|x)\}.$$
The equality in \eqref{step1show} can be expressed as
\begin{align}
\max_{(a_1,a_2)\in\mathcal{A}}\min_{\lambda \geq 0} [a_2 + \lambda(R - a_1) ]&=\min_{\lambda \geq 0}\max_{(a_1,a_2)\in\mathcal{A}} [a_2 + \lambda(R - a_1) ].  \end{align}
We use Sion's Minimax Theorem \cite{sion1958} to show that we can exchange the order of $\max$ and  $\min$. We can use Sion's minimax theorem because
\begin{itemize}
    \item The function $a_2 + \lambda(R - a_1)$ is linear in $a_1$, $a_2$, and $\lambda$.
    \item The domain set  $\{\lambda: \lambda \geq 0 \}$ is a convex set. 
    \item The set $\mathcal{A}$ is a closed, compact and convex set. To see this, observe that the set $\mathcal{A}$ is compact because
$a_1\in[0,H(X|Z)]$ and $a_2\in[0,H(Y|Z)]$; it is
closed because one can impose a cardinality bound of $|\mathcal{U}|\leq |\mathcal{X}|+2$ on the auxiliary random variable $U$ using the standard arguments based on the Caratheodery's theorem. The set $\mathcal{A}$ is convex because for any $p_{U_1|X}$ and $p_{U_2|X}$, we can take $Q\in\{1,2\}$ to be a Bernouli (1/2) random variable that is independent of all previously defined variables and define $U=(Q,U_Q)$. 

\end{itemize}

\end{proof}

\section{Evaluation of the upper bound in the Gaussian setting}

\label{appGaussian}
Consider the following optimization problem
 \begin{align*}
G_{P_{XY},Q_{XY},P_{XZ},Q_{XZ}}(R)=\sup_{P_{U|X}} \mathbb{E}_{P_U}\left[\max_{\hat Q_X}D(P_{Y|U}\|\hat Q_{Y})-D(P_{Z|U}\|\hat Q_Z)\right]
 \end{align*} 
 where the supremum is over $P_{U|X}$ satisfying  $I_P(U;X|Y)\leq R$ and $I_P(U;X|Z)\leq R$. 
Assume that $X$ is a scalar standard Gaussian random variable under both hypotheses. Assume that $P_{Y|X}, Q_{Y|X}, P_{Z|X}$ and $Q_{Z|X}$ are additive Gaussian noise channels with additive noise variances $\sigma_{P,Y|X}^2$, $\sigma_{Q,Y|X}^2$, $\sigma_{P,Z|X}^2$, $\sigma_{Q,Z|X}^2$ respectively.
If $\sigma_{Q,Z|X}>\sigma_{Q,Y|X}$, the above expression becomes infinity. This can be seen by choosing a constant random variable $U$ and following the choice in Example \ref{exampleinf} given in Section \ref{sec31}.

In the following, we assume that $\sigma_{Q,Z|X}\leq \sigma_{Q,Y|X}$. In this case, we can find a jointly Gaussian distribution $Q_{XYZ}$ (a coupling) such that
$$Q_{XYZ}=Q_{X}Q_{Z|X}Q_{Y|Z}.$$
Note that the assumption $Q_{XYZ}=Q_{X}Q_{Z|X}Q_{Y|Z}$ would then imply that the inner maximum over $\hat Q_X$ is bounded for every $P_{U|X}$. Since 
\begin{align*}
&\mathbb{E}_{P_U}\left[\max_{\hat Q_X}D(P_{Y|U}\|\hat Q_{Y})-D(P_{Z|U}\|\hat Q_Z)\right] \\
&=-h_P(Y|U)+h_P(Z|U)+
\mathbb{E}_{P_{U}}\max_{\hat Q_X}\mathbb{E}_{P_{YZ|U}}\log\frac{\hat Q_Z}{\hat Q_Y},
\end{align*}
we can observe that for every $P_{UX}$ we have
\begin{align*}
     &\mathbb{E}_{P_{U}}\max_{\hat Q_X}\mathbb{E}_{P_{YZ|U}}\log\frac{\hat Q_Z}{\hat Q_Y}\leq \mathbb{E}_{P_{U}}\max_{\hat Q_X}\mathbb{E}_{P_{YZ|U}}\log\frac{\hat Q_Z}{\hat Q_{YZ}} \\
     &=\mathbb{E}_{P_{U}}\max_{\hat Q_X}\mathbb{E}_{P_{YZ|U}}\log\frac{1}{\hat Q_{Y|Z}}=\mathbb{E}_{P_{U}}\mathbb{E}_{P_{YZ|U}}\log\frac{1}{ Q_{Y|Z}}<\infty.
 \end{align*}

 For the setting $\sigma_{Q,Z|X}\leq \sigma_{Q,Y|X}$, we would like to prove that there is a maximizer $P_{U|X}$ in which $U$ and $X$ are jointly Gaussian random variables.

 Let 
 \begin{align*}
     \sigma^2_X&=\mathsf{Var}_P[X],
    \\ \sigma^2_{X|Y}&=\mathbb{E}[\mathsf{Var}_P[X|Y]],
    \\ \sigma^2_{X|Z}&=\mathbb{E}[\mathsf{Var}_P[X|Z]]. 
 \end{align*}
Then, using the fact that $(X,Y,Z)$ is jointly Gaussian under $P$, we can write the constraints as
\begin{align*}
    h(X|YU)&\geq h(X|Y)-R= \frac{1}{2} \log(2\pi e\sigma_{X|Y}^2)-R,
\\
    h(X|ZU)&\geq h(X|Z)-R= \frac{1}{2} \log(2\pi e\sigma_{X|Z}^2)-R.
\end{align*}
Consider the following  optimization problem where we relax the Gaussian assumption on the marginal distribution of $X$ under $P$:
\begin{align*}
\sup_{P_{UX}} \mathbb{E}_{P_U}\left[\max_{\hat Q_X}D(P_{Y|U}\|\hat Q_{Y})-D(P_{Z|U}\|\hat Q_Z)\right]
 \end{align*} 
 where the supremum is over $P_{UX}$ satisfying 
\begin{align}
    h(X|YU)&\geq  \frac{1}{2} \log(2\pi e\sigma_{X|Y}^2)-R,\label{eqnConstrx1}
\\
    h(X|ZU)&\geq  \frac{1}{2} \log(2\pi e\sigma_{X|Z}^2)-R,\label{eqnConstrx2}
\end{align}
 and satisfying the following covariance constraint:
 \begin{align}
     \mathsf{Var}_P[X]\leq \sigma_X^2.\label{eqnConstrx3}
 \end{align}
In the above relaxed problem, the channels $P_{YZ|X}$ and $Q_{YZ|X}$ are kept as the same Gaussian channels.
The relaxed problem is similar to channel coding type optimizations, as the distribution on $X$ is arbitrary. The advantage of the relaxation is that we can use Caratheodery's theorem to restrict the alphabet size of random variable $U$ to $|\mathcal U|\leq 3$. This restriction is crucial when we employ the approach in \cite{gen14} to prove the optimality of the Gaussian channel $P_{U|X}$. Herein, we also utilize the idea of "perturbated expressions" introduced in \cite{gon21}.

Let $G$ be a standard Gaussian random variable, independent of $(X,Y,Z)$, and (under either hypothesis) let  
\begin{align*}
\tilde X&=X+\delta\cdot G  \\
\tilde Y&=Y+\delta\cdot G  \\
\tilde Z&=Z+\delta\cdot G 
\end{align*}
where $\delta>0$ is a small real number. Observe that the following Markov chain holds for any $P_{UX}$:
$$U\rightarrow X\rightarrow \tilde X\rightarrow (\tilde Y,\tilde Z).$$
Consider the following perturbed optimization problem:
 \begin{align}
&\sup_{P_{UX}}
\mathbb{E}_{P_U}\left[\max_{\hat Q_X}D(P_{\tilde Y|U}\|\hat Q_{Y})-D(P_{\tilde Z|U}\|\hat Q_Z)\right] \nonumber
\\&-\delta\cdot h_P(\tilde Y|U)
% \\&=
%\sup_{P_{U|X}} -h_P(\tilde Y|U)+h_P(\tilde Z|U)+
%\mathbb{E}_{P_{U}}\max_{\hat Q_X}\mathbb{E}_{P_{\tilde Y\tilde Z|U}}\log\frac{\hat Q_Z}{\hat Q_Y}+\delta I(U;\tilde Y)\label{perturbedopti}
 \end{align} 
 where the supremum is over $P_{UX}$ satisfying   
 \begin{align}
    h_P(\tilde X|\tilde YU)&\geq  \frac{1}{2} \log(2\pi e\sigma_{X|Y}^2)-R,\label{eqnConstrx1n}
\\
    h_P(\tilde X|\tilde ZU)&\geq  \frac{1}{2} \log(2\pi e\sigma_{X|Z}^2)-R.\label{eqnConstrx2n}
\end{align}
 and \eqref{eqnConstrx3}.  Letting $\delta$ go to zero yields the original optimization problem. Therefore, it suffices to show the existence of a maximizer in which $U$ and $X$ are jointly Gaussian for the perturbed problem for any $\delta>0$. 
 
Let $p^*(u,x, \tilde x, \tilde y, \tilde z)$ be a maximizer of the optimization problem.\footnote{Existence of a maximizer follows from the standard arguments in Appendix II of \cite{gen14}. The Gaussian noise $G$ is considered to make sure that $\tilde{X}$ has a bounded density in order to get the proper convergence results. Alternatively, one can also utilize the result in
\cite{mahvari2023stability} which does not necessitate showing the existence of a maximizer. }
Let $$(U^*_1,X^*_{1},{\tilde X}^*_{1}, \tilde Y^*_1, \tilde Z^*_1)$$ and $$(U^*_2,X^*_{2},{\tilde X}^*_{2}, \tilde Y^*_2, \tilde Z^*_2)$$ to be two independent copies of a maximizer $p^*(u,x, \tilde x, \tilde y, \tilde z)$ of the optimization problem
subject to the constraints in the problem. Further, let $V$ denote the optimal value.
Let
\[
\begin{aligned}
&(U', X', \tilde X',\tilde Y',\tilde Z') := (( U^*_1, U^*_2), ( X^*_1, X^*_2), ( \tilde X^*_1, \tilde X^*_2),( \tilde Y^*_1, \tilde Y^*_2), (\tilde Z^*_1, \tilde Z^*_2)) \\
\end{aligned}
\]
and
\[
\begin{aligned}
\begin{pmatrix} X_{+} \\ X_{-} \end{pmatrix}
:= \begin{pmatrix} \sqrt{t} & \sqrt{1-t} \\ -\sqrt{1-t} & \sqrt{t} \end{pmatrix} \begin{pmatrix} X^*_{1} \\ X^*_{2} \end{pmatrix} \end{aligned}
\]
where $t = \frac{1}{2}$. We define channel outputs $\tilde X_{+}, \tilde X_{-}, \allowbreak \tilde Y_{+}, \tilde Y_{-}, \tilde Z_{+}, \tilde Z_{-}$ in the standard manner as in \cite{gen14}.

Set
\begin{align}
\hat U_- &:=(U',\tilde Y_{+}) &\quad   \hat U_+ &:=(U', \tilde Z_{-}),  
\end{align}

Since rotations preserve mutual information, note that standard manipulations yield
\begin{align}
&2V \nonumber\\
&=\mathbb{E}_{P_{U'}}\Big[\max_{\hat Q_{X_+X_-}}D(P_{\tilde Y_{+}\tilde Y_{-}|U'}\|\hat Q_{Y_+Y_-})
-D(P_{\tilde Z_{+}\tilde Z_{-}|U'}\|\hat Q_{Z_+Z_-})\Big] \nonumber 
-\delta h_P(\tilde Y_{+}\tilde Y_{-}|U')\nonumber
\\&=\mathbb{E}_{P_{U'}}\bigg[\max_{\hat Q_{X_+X_-}}D(P_{\tilde Y_{+}\tilde Y_{-}|U'}\|\hat Q_{\tilde Y_{+}\tilde Y_{-}})-D(P_{\tilde Y_{+}\tilde Z_{-}|U'}\|\hat Q_{\tilde Y_{+}\tilde Z_{-}})\nonumber
\\&\quad+D(P_{\tilde Y_{+}\tilde Z_{-}|U'}\|\hat Q_{\tilde Y_{+}\tilde Z_{-}})-D(P_{\tilde Z_{+}\tilde Z_{-}|U'}\|\hat Q_{\tilde Z_{+}\tilde Z_{-}})\bigg]\nonumber
\\&\quad 
-\delta h_P(\tilde Y_{+}|U')-\delta h_P(\tilde Y_{-}|\tilde Y_{+}U')
\nonumber
\\&\leq \mathbb{E}_{P_{U'}}\bigg[\max_{\hat Q_{X_+X_-}}D(P_{\tilde Y_{-}|U'\tilde Y_{+}}\|\hat Q_{\tilde Y_{-}|\tilde Y_+}|P_{\tilde Y_{+}|U'})
 -D(P_{\tilde Z_{-}|U'\tilde Y_{+}}\|\hat Q_{\tilde Z_{-}|\tilde Y_+}|P_{\tilde Y_{+}|U'})\bigg]\nonumber
\\&\quad+\mathbb{E}_{P_{U'}}\bigg[\max_{\hat Q_{X_+X_-}} D(P_{\tilde Y_{+}|U'\tilde Z_{-}}\|\hat Q_{\tilde Y_{+}|\tilde Z_{-}}|P_{\tilde Z_{-}|U'})
-D(P_{\tilde Z_{+}|\tilde Z_{-}U'}\|\hat Q_{\tilde Z_{+}|\tilde Z_{-}}|P_{\tilde Z_{-}|U'})\bigg]\nonumber
\\&\quad -\delta h_P(\tilde Y_{+}|U'\tilde{Z}_-)-\delta h_P(\tilde Y_{-}|\tilde Y_{+}U')-\delta I_P(\tilde{Y}_+;\tilde{Z}_-|U')\nonumber
\\&\leq \mathbb{E}_{P_{U'\tilde Y_{+}}}\bigg[\max_{\hat Q_{X_-}}D(P_{\tilde Y_{-}|U'\tilde Y_{+}}\|\hat Q_{\tilde Y_{-}})-D(P_{\tilde Z_{-}|U'\tilde Y_{+}}\|\hat Q_{\tilde Z_{-}})\bigg]\nonumber
\\&\quad+\mathbb{E}_{P_{U'\tilde Z_{-}}}\bigg[\max_{\hat Q_{X_+}}D(P_{\tilde Y_{+}|U'\tilde Z_{-}}\|\hat Q_{\tilde Y_{+}})-D(P_{\tilde Z_{+}|\tilde Z_{-}U'}\|\hat Q_{\tilde Z_{+}})\bigg]\nonumber
\\&\quad -\delta h_P(\tilde Y_{+}|U'\tilde{Z}_-)-\delta h_P(\tilde Y_{-}|\tilde Y_{+}U')-\delta I_P(\tilde{Y}_+;\tilde{Z}_-|U')\nonumber
\\&=\mathbb{E}_{P_{\hat U'_-}}\bigg[\max_{\hat Q_{X_-}}D(P_{\tilde Y_{-}|\hat U'_-}\|\hat Q_{\tilde Y_{-}})-D(P_{\tilde Z_{-}|\hat U'_-}\|\hat Q_{\tilde Z_{-}})\bigg]\nonumber
\\&\quad+\mathbb{E}_{P_{\hat U'_+}}\bigg[\max_{\hat Q_{X_+}}D(P_{\tilde Y_{+}|\hat U'_+}\|\hat Q_{\tilde Y_{+}})-D(P_{\tilde Z_{+}|\hat U'_+}\|\hat Q_{\tilde Z_{+}})\bigg]\nonumber
\\&\quad 
-\delta h_P(\tilde Y_{+}|U'_+)-\delta h_P(\tilde Y_{-}|U'_-)-\delta I_P(\tilde{Y}_+;\tilde{Z}_-|U').
\label{leqndiff}
\end{align}
Moreover,
 \begin{align*}
    &h_P(\tilde X_+\tilde X_-|\tilde Y_+\tilde Y_- U')
    \\&=
    h_P(\tilde X_-|\tilde Y_- U'\tilde Y_+)+h_P(\tilde X_+|\tilde X_-\tilde Y_+\tilde Y_- U')
    \\
    &=
    h_P(\tilde X_-|\tilde Y_- U'\tilde Y_+)+h_P(\tilde X_+|\tilde X_-\tilde Z_-\tilde Y_+\tilde Y_- U')
   \\ &\leq 
    h_P(\tilde X_-|\tilde Y_- U'\tilde Y_+)+h_P(\tilde X_+|\tilde Y_+U'\tilde Z_-)
    \\ &=
    h_P(\tilde X_-|\tilde Y_- \hat{U}_-)+h_P(\tilde X_+|\tilde Y_+\hat{U}_+).
\end{align*}
Similarly,
 \begin{align*}
    h_P(\tilde X_+\tilde X_-|\tilde Z_+\tilde Z_- U')&\leq 
    h_P(\tilde X_-|\tilde Z_-\hat U_-)+h_P(\tilde X_+|\tilde Z_+\hat U_+).
\end{align*}
Now, assume that $S$ takes values in $\{+,-\}$ with uniform probability, and is independent of all previously defined variables. 
When $S=+$, we set
\begin{align*}
    &U''=(S, \hat{U}_+),~~~ X''=X_+,~~~ \tilde{X}''=\tilde X_+,
    \\&\tilde{Y}''=\tilde Y_+,~~~ \tilde{Z}''=\tilde Z_+.
\end{align*}
When $S=-$, we set
\begin{align*}
    &U''=(S, \hat{U}_-),~~~ X''=X_-,~~~ \tilde{X}''=\tilde X_-,
    \\&\tilde{Y}''=\tilde Y_-,~~~ \tilde{Z}''=\tilde Z_-.
\end{align*}
Notice that
 the distribution on $U''X''$ is a candidate maximizer of the expression and hence must induce a value of at most $V$. Thus, from \eqref{leqndiff} we obtain that 
$$I_P(\tilde Y_{+};\tilde Z_{-}|U')=0.$$
The rest of the arguments mimic ones in \cite{gen14} in Propositions 2, 8, and Corollary 3 of \cite{gen14}, and we omit the details.
 Using Proposition 2 in \cite{gen14}, we deduce that $I_P(\tilde X_{-};\tilde X_{+}|U')=0$. By the Skitovic-Darmois characterization of Gaussians, we deduce that
 conditioned on $U'$, random variables $X^*_{1}, X^*_{2}$ are also jointly Gaussians, and they have the same covariance matrix (independent of $U'$). 
Thus, a Gaussian $p(x|u)$ with the same covariance for each $u$ is an optimal solution of the optimization
problem.

\section*{Acknowledgements}
This work was supported by the CUHK Direct Grant 4055193. The authors thank Chandra Nair for inspiring discussions and for helping with the proof of the optimality of Gaussian distributions.

\bibliographystyle{ieeetr}
\bibliography{references}

\end{document}